\documentclass[conference]{IEEEtran}
\IEEEoverridecommandlockouts
\usepackage{cite}
\usepackage{amsmath,amssymb,amsfonts}
\usepackage{graphicx}
\usepackage{textcomp}
\usepackage{xcolor}
\def\BibTeX{{\rm B\kern-.05em{\sc i\kern-.025em b}\kern-.08em
    T\kern-.1667em\lower.7ex\hbox{E}\kern-.125emX}}

\usepackage[utf8]{inputenc}
\usepackage{bbm}
\usepackage{graphicx}
\usepackage{hyperref}
\usepackage{subcaption}
\usepackage{amsmath}
\usepackage{amsfonts}
\usepackage{amsthm}
\usepackage{listings}
\usepackage{color}
\usepackage{microtype}
\usepackage{cleveref}
\usepackage{algorithm}
\usepackage{enumitem}
\usepackage{xspace}
\usepackage[noend]{algpseudocode}

\usepackage[subtle]{savetrees}

\usepackage{booktabs}
\usepackage{tabularray}

\newcommand{\sysname}{\textsc{OReO}\xspace}
\newcommand{\reorg}{\textsc{Reorganizer}\xspace}
\newcommand{\layout}{\textsc{Layout Manager}\xspace}

\newcommand{\anonsys}{SuperCollider\xspace}
\newcommand{\dumts}{\textsc{D-UMTS}\xspace}
\newcommand{\minihead}[1]{{\vspace{.45em}\noindent\textbf{#1.} }}

\newtheorem*{theorem*}{Theorem}
\newtheorem{theorem}{Theorem}[section]

\newtheorem{remark}{Remark}[section]

\newcommand{\squishitemize}{
 \begin{list}{$\bullet$}
  { \setlength{\itemsep}{0pt}
     \setlength{\parsep}{0pt}
     \setlength{\topsep}{0pt}
     \setlength{\partopsep}{0pt}
     \setlength{\leftmargin}{1.95em}
     \setlength{\labelwidth}{1.5em}
     \setlength{\labelsep}{0.5em} } }

\newcounter{Lcount}
\newcommand{\squishlist}{
    \begin{list}{\arabic{Lcount}. }
   { \usecounter{Lcount}
        \setlength{\itemsep}{0pt}
        \setlength{\parsep}{3pt}
        \setlength{\topsep}{0pt}
        \setlength{\partopsep}{0pt}
        \setlength{\leftmargin}{2em}
        \setlength{\labelwidth}{1.5em}
        \setlength{\labelsep}{0.5em} } }

\newcommand{\squishend}{\end{list}}

\newcommand{\etal}{et al.\ }

\newcommand{\E}{\mathbf{E}}

\newcommand{\cA}{\mathcal{A}}

\newcommand{\cD}{\mathcal{D}}

\newcommand{\cH}{{\mathcal{H}}}

\newcommand{\cQ}{\mathcal{Q}}

\newcommand{\cS}{\mathcal{S}}
\newcommand{\cC}{\mathcal{C}}

\newcommand{\revision}[1]{{\color{black} #1}\xspace}

\begin{document}

\title{Dynamic Data Layout Optimization \\ with Worst-case Guarantees 
% \\
%\thanks{Research reported in this work was supported by an Amazon Research Award Fall 2023. Any opinions, findings, and conclusions or recommendations expressed in this material are those of the author(s) and do not reflect the views of Amazon.}
}

\author{\IEEEauthorblockN{Kexin Rong}
\IEEEauthorblockA{\textit{Georgia Institute of Technology,}\\ \textit{VMware Research} \\ krong@gatech.edu}
\and
\IEEEauthorblockN{Paul Liu}
\IEEEauthorblockA{\textit{Stanford University} \\
paul.liu@stanford.edu}
\and 
\IEEEauthorblockN{Sarah Ashok Sonje}
\IEEEauthorblockA{\textit{Georgia Institute of Technology} \\
sarah.sonje@gatech.edu}
\and
\IEEEauthorblockN{Moses Charikar}
\IEEEauthorblockA{\textit{Stanford University} \\
moses@cs.stanford.edu}
}

\maketitle

\begin{abstract}
Many data analytics systems store and process large datasets in partitions containing millions of rows. By mapping rows to partitions in an optimized way, it is possible to improve query performance by skipping over large numbers of irrelevant partitions during query processing. This mapping is referred to as a data layout. Recent works have shown that customizing the data layout to the anticipated query workload greatly improves query performance, but the performance benefits may disappear if the workload changes. Reorganizing data layouts to accommodate workload drift can resolve this issue, but reorganization costs could exceed query savings if not done carefully.

In this paper, we present an algorithmic framework \sysname that makes online reorganization decisions to balance the benefits of improved query performance with the costs of reorganization. Our framework extends results from Metrical Task Systems to provide a tight bound on the worst-case performance guarantee for online reorganization, without prior knowledge of the query workload. Through evaluation on real-world datasets and query workloads, our experiments demonstrate that online reorganization with \sysname can lead to an up to 32\% improvement in combined query and reorganization time compared to using a single, optimized data layout for the entire workload.
\end{abstract}

\begin{IEEEkeywords}
data layout optimization, metrical task systems
\end{IEEEkeywords}

\section{Introduction}
To keep up with increasingly large data demands, modern data analytics systems partition data in chunks containing millions of records, which are then compressed and persisted in cost-effective storage options such as Amazon S3. These partitions are often the smallest unit for I/O operations, meaning that all partitions containing relevant data must be accessed in their entirety during query processing. 

The mapping from individual data records to different partitions, known as the \emph{data layout}, can have a significant impact on query performance. Query optimizers utilize partition-level metadata, such as the min-max ranges of each column, to determine which partitions can be skipped during query processing, thereby enhancing performance~\cite{graefe2009fast,oraclezm,sparkskipping}. By default, many systems simply partition the dataset according to one or more predefined sort columns, such as the arrival time of data records~\cite{sparksql,borthakur2011apache}. However, this means that queries unrelated to the arrival time may still need to access a large number of partitions.
Recent research shows that customizing the data layout to specific query workloads can significantly improve data-skipping performance~\cite{sun2014fine,yang2020qd,ding2021instance}. For example, Qd-tree~\cite{yang2020qd} uses query predicates to partition the dataset in a way that maximizes partitions skipped for the given query workload. 

%Recently, researchers have explored how to use knowledge of the query workload to improve data layout designs~\cite{sun2014fine,yang2020qd,ding2021instance}. These workload-aware layout designs have achieved state-of-the-art data skipping performances compared to traditional designs that do not make use of such information. 

\begin{figure}[t]
    \centering
    \includegraphics[width=\linewidth]{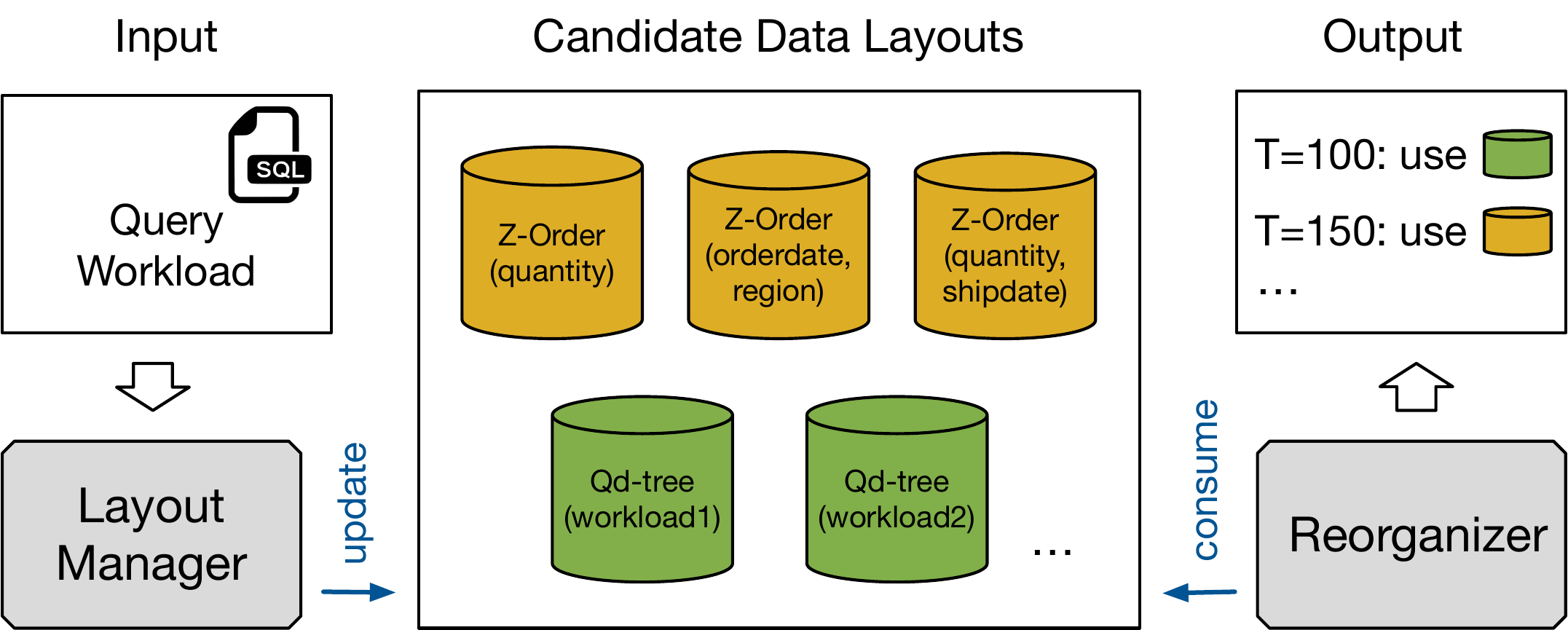}
    \caption{Overview of \sysname, an algorithmic framework for making online data reorganization decisions to minimize the total costs of query processing and data reorganization over an unknown query sequence. }
    \label{fig:overview}
\end{figure}

However, the performance of these workload-aware layouts can degrade significantly when the target query workload drifts. Reorganizing the data, or switching to a different data layout, to adapt to workload changes requires accessing and rewriting large amounts of data, and could be 10$\times$ to 100$\times$ more expensive than running a full scan query on the table. It may be worth incurring this cost if the new layout leads to future query cost savings that outweigh the upfront cost of reorganization. Since users pay for the cost incurred by both query processing and reorganization,  data analytics systems need to consider both costs instead of only focusing on reducing the query costs. In fact, rule-based reorganization heuristics have already been implemented in systems such as Snowflake\revision{~\cite{snowauto}} and Delta Lake\revision{~\cite{deltaauto}} to continuously optimize the organization of datasets during query processing. However, the heuristics are often set empirically without any guarantees of performance, and it is not clear how these systems can reason about the performance implications of their reorganization decisions.
%Therefore, data analytics systems must constantly decide whether or not to reorganize, and if so, which data layout to switch to during query processing.

In this work, we introduce \sysname (Online Re-organization Optimizer), an algorithmic framework that outputs a data reorganization schedule to minimize the total query and reorganization costs over an unknown query sequence, while providing provable worst-case performance guarantees (Figure~\ref{fig:overview}). Since it may not always be possible to access historical workloads, such as with private customer datasets, our approach explores the ``online" extreme on the design spectrum and assumes that {\em no prior knowledge of the query workload} is available. This forces us to design the algorithm to react to the query stream on the fly. Despite this, we show that our algorithm can perform favorably compared to algorithms that have access to the entire workload in advance. To achieve this, \sysname leverages recent advances in workload-aware data layouts to generate new candidate layouts during query processing, and makes reorganization decisions by adapting classical results from the study of Metrical Task Systems (MTS) in online learning~\cite{BLS92}. \sysname provides guarantees in the form of {\em competitive ratios}, which is the total cost of online reorganization divided by that of an optimal offline solution.
%(a single data layout optimized for the entire workload). %The competitiveness is bounded against an {\em oblivious} adversary, who has access to the algorithm but not the randomized results of the algorithm. 
%In order to make online algorithms work with workload-dependent data layouts, the core contribution in \sysname is to separate two previously related concerns in the problem: which data layouts to use and how to switch between them.  %We adopt the classic oblivious adversary model, meaning that an adversary must prepare the entire sequence for queries and state modifications in advance.

A key technical contribution lies in the ability of \sysname to decouple the process of generating data layouts from the process of making reorganization decisions. This separation is necessary because, in classic Metrical Task Systems (MTS), the system operates in a fixed set of states, equivalent to the set of candidate data layouts in our application. The total number of states directly affects the quality of the online solution, with a smaller state space leading to a better competitive ratio. However, the state space for online reorganization is prohibitively large and intractable to work with, as it includes all possible partitionings of a dataset. With prior knowledge of the workload, we can reduce the size of the state space by precomputing a small set of good-performing data layouts using workload-aware layout designs. Without such knowledge, the decision maker must adapt as the system observes more queries, adding better-performing data layouts to the state space incrementally. This requires the state space to change over the course of query processing, instead of remaining fixed as in traditional Metrical Task Systems.
%; processing the same task incurs different costs in different system states and there is a cost associated with switching states. The quality of the online solution is directly related to the number of total system states: the smaller the state space, the tighter the competitive ratio. 

%Ideally, the decision maker only needs to consider a small number of data layouts that subsume the rest in performances, which can be computed for given workloads using workload-dependent layout designs. It seems that we either need to know the workload in advance to precompute a compact state space, or are left with an intractable number of states for the online algorithm to switch between. 
%that making reorganization decisions on how to switch between system states does not require knowledge of the workload, but precomputing a fixed, reduced state space does. Specifically, it is intractable to enumerate the state space containing all possible partitionings of the dataset. 

To address this challenge, we propose a novel, \emph{dynamic} variant of the classic uniform metrical task system problem, which permits arbitrary modification of the state space during query processing. The modifications are modeled as state management queries, which can add and remove any system state at any time during query processing. This new framework allows us to separate the state generation and state transition into two independent components: the \reorg and the \layout. The \reorg adapts a classic randomized algorithm for uniform metrical task systems to support dynamic state spaces and achieves a provably tight competitive ratio, similar to guarantees provided in the original setup. The \layout, on the other hand, is responsible for updating the state space. It starts with a default layout such as partitioning by time, and generates new layouts tailored to the current workloads on-the-fly as the system observes more queries. The \layout also takes into account the diversity of the layouts generated to prune redundant states and keep the state space compact. Together, the framework offers a systematic approach for balancing query improvements and reorganization costs over an unknown query sequence. 

%We show that the adapted algorithm provides similar guarantees to the original one, except relating performance to the size of the dynamic state space instead of the entire state space. 
In summary, this paper makes the following contributions: 
\begin{itemize}
    \item A formal online algorithm framework, \sysname, to model the online layout optimization problem. The problem setup differs from traditional online learning literature in that it does not assume that all the states are known apriori.
    \item An MTS algorithm that achieves an asymptotically tight competitive ratio under dynamically changing states. 
    \item Evaluation of the framework on several real-world datasets and workloads. Our results show that online reorganization without workload knowledge improves upon an offline layout precomputed for the entire workload by up to 32\% in end-to-end compute time. 
\end{itemize}

%The rest of the paper proceeds as follows. Section~\ref{sec:background} discusses current practices and motivations for dynamic reorganization and provides background on metrical task systems. Section~\ref{sec:overview} presents the problem formulation and the overview of \sysname. Sections~\ref{sec:reorg} and \ref{sec:layout} introduce the two main components of the framework: the MTS-based \reorg and the \layout. Section~\ref{sec:eval} presents a detailed evaluation on the performance of the framework and its sensitivity to changes in main parameters. Section~\ref{sec:related} reviews related work in layout optimization and adaptive physical design and tuning. Section~\ref{sec:discuss} discusses directions for future work and Section~\ref{sec:conclude} concludes.
\section{Background and Motivation}
\label{sec:background}

We start this section with a discussion of current industry practices for dynamic data layout optimization, which inspire the design of \sysname (\S~\ref{sec:scenario}). We then provide background on classic results from Metrical Task Systems, which serve as the theoretical foundation of \sysname (\S~\ref{sec:mts}). 

%We then illustrate the potential performance benefits of dynamically switching data layouts compared to using a static, optimized data layout for the entire query workload through a synthetic example (\S~\ref{sec:synexample}). 

\subsection{Motivation: Dynamic Layout Optimization in Practice}
\label{sec:scenario}

%By organizing a dataset based on commonly used filters in queries, an optimized data layout can significantly reduce the amount of data that needs to be read and improve query performance. Therefore, data layout optimization is often used as a knob to help troubleshoot the performance of data analytics systems. As these systems are increasingly offered as a service, practitioners are more and more interested in automating the tuning of such performance-related features.

Current systems often implement dynamic data layout optimization policies via predefined rules or thresholds. These heuristics, simple to implement and maintain, are particularly useful when historical workload data is not accessible, such as with private customer datasets. For example, Delta Lake provides a layout optimization feature based on Z-Order that can be activated manually or automatically when the number of small files exceeds a threshold during table creation and merges\revision{~\cite{deltaauto}}. Snowflake has a rule-based automatic clustering feature that reorganizes data automatically when more than a certain number of data partitions overlap in a specified column\revision{~\cite{snow}}. Notably, users incur charges based on the computational resources used during this reorganization when the automatic clustering is enabled. Therefore, it is crucial to balance the costs of query processing with those of reorganization in these systems. Both Delta Lake and Snowflake perform reorganization in the background through a separate process, thereby minimizing its impact on query performance.
%Users are charged based on the computation resources consumed during the reorganization process when the automatic clustering feature is enabled. Therefore, it is equally important to consider the cost incurred by query processing and the costs incurred by reorganization in these systems. In both cases, the reorganization is done in the background using a separate process, which minimizes the impact on query performance.

\sysname draws inspiration from these current rule-based reorganization practices and aims to provide a more formal framework for considering the trade-off between query costs and reorganization costs. Similar to rule-based heuristics, \sysname responds to the query stream and makes reorganization decisions on-the-fly. However, unlike heuristic methods which are often set empirically, \sysname provides a provable worst-case performance guarantee over all possible input sequences. 

%In addition to query costs, it is equally important to consider costs incurred by reorganization in these scenarios, since users would pay for both types of costs. For example, Snowflake charges users based on the computation resources consumed during the reorganization process when the automatic clustering feature is enabled.

%Compared to learning-based solutions that make such decisions based on the prediction of future query patterns, \sysname is easier to integrate and maintain since there is no need for training with historical data or updates when the target distribution changes after deployment.

\subsection{Background: Metrical Task Systems}
\label{sec:mts}
\sysname builds upon classic results from Metrical Task Systems (MTS)~\cite{BLS92}. In MTS, the decision maker controls a system that can be in one of $n$ system states. The states are associated with a distance metric (i.e., obeying the triangle inequality), which defines the movement cost of switching system states. The system takes as input an unknown sequence of tasks, where each task incurs a (possibly different) service cost in each system state. For each new task, the system can choose to service it in the current state or to move to a different state and then service it by paying an additional movement cost. The goal of the decision maker is to minimize the sum of service costs and movement costs over the task sequence which is processed one-by-one, without knowledge of the future.
 
The main benchmark of comparison for MTS is the \emph{competitive ratio}. The competitive ratio is the worst case over all possible input sequences of the performance of the online algorithm, which observes one new task at a time, divided by the performance of the optimal offline algorithm that is presented with the entire task sequence in advance. Several algorithms with varying competitive ratios for MTS exist, each with different assumptions on the state space. The tightness of the bound is typically related to the size of the state space~\cite{borodin2005online, BLS92, IS98}. 

The key difference between our setup and the traditional MTS problem is that MTS works with a fixed set of states that is known to the decision maker apriori, whereas we would like to allow the state space to evolve. Our work introduces a novel MTS problem formulation that addresses this difference and achieves a tight competitive ratio that is logarithmic in the maximum size of the dynamic state space. The results are therefore of interest beyond the scope of the online layout optimization problem we study in this paper.

\section{Overview}
\label{sec:overview}
In this section, we present the problem formulation (\S~\ref{sec:problem}), an overview of \sysname's components and workflow (\S~\ref{sec:algoverview}), \revision{as well as the main assumptions and limitations (\S~\ref{sec:limitations})}.
%an overview of the algorithmic framework for the online layout optimization problem (\S~\ref{sec:algoverview}) and discuss the main assumptions and limitations of the framework (\S~\ref{sec:assum}). 

\subsection{Problem Formulation}
\label{sec:problem}
We formulate the online layout optimization problem as a variant of uniform metrical task systems (UMTS) studied in the literature. We are given a set of states $\cS$ and an ordered stream of queries $\cQ$. For each $s \in \cS$ and $q \in \cQ$, there is some cost $c(s, q) \in [0, 1]$ of servicing the query $q$ in state $s$. To switch between states in $\cS$, there is a reorganization cost of $\alpha > 1$. The goal is to service all queries in $\cQ$ while minimizing the sum of query costs and reorganization costs over the entire sequence, reflecting the total compute costs users would pay to use systems like cloud data warehouses.

To model the introduction of new data layouts while processing the query stream, we propose a novel {\em dynamic} variant of UMTS, \dumts, where the set of states $\cS$ is allowed to vary over time. In \dumts, we permit arbitrary addition and removal of states from the state space $\cS$ in the midst of query processing, thus making the set of states dynamic. The flexibility of removing states allows for controlling the size of the state space in our applications. 
%In D-UMTS, we consider switching between a set of data layouts with the same reorganization cost, and we discuss the implication of this assumption in more detail in the next section. 
In the rest of the text, we use states and data layouts interchangeably in the context of describing our framework.

%We assume that switching between data layouts (or states) requires reconstructing the layout from scratch for the entire dataset. As such, switching between layouts can be modeled by a constant cost. In other words, the metric between states is simply the uniform metric. 

%We discuss the implication of this assumption in more detail in the next section.
%\footnote{Although there exists layout schemes that allow for partial reconstruction, assuming constant reconstruction costs allows us to model a broader range of layout algorithms.}
% In addition, we generate data layouts and make reorganization decisions for each table independently, using queries over that table. We support arbitrary query predicates with conjunctions, disjunctions and negations. Queries with key--foreign key joins are supported over the denormalized table. 

%Specifically, our cost model considers the total compute time incurred by querying and partitioning the datasets. 
We model the query and reorganization costs as the total compute time incurred by querying and partitioning the datasets. To estimate the query cost, we use the fraction of the dataset accessed by each query, which has been shown to be a reliable proxy for query performance~\cite{yang2020qd, rong2020approximate}. To estimate the reorganization cost, we represent the relative overhead of reorganizing (which includes compressing and writing partitions) compared to querying (mostly reading partitions) via a parameter $\alpha$ in the cost model. $\alpha$ is the {\em expected} ratio between the compute time spent on reorganization compared to the time spent on a full table scan query ($\alpha := \mathbf{E}[\frac{\text{reorganization time}}{\text{full table scan}}]$). The value of $\alpha$ varies based on system configurations and can be measured experimentally. In our experiments, $\alpha$ typically falls in the range of $60\times$ to $100\times$, similar to values reported in previous work~\cite{ding2021instance}.

We analyze the worst-case performance of the framework using the classic \emph{oblivious} adversary model~\cite{borodin2005online}, meaning that an adversary must prepare the entire sequence for queries and state modifications in advance. Since the set of states is constantly changing through state modification requests, the adversary must additionally use the same set of states available to our algorithm at all times. Note that oblivious adversaries are quite powerful, as they necessarily force our algorithm to be randomized. If our algorithm is deterministic, the adversary can simulate our algorithm on its input in advance, and again prepare a sequence to remove the exact state that the system is currently in. Finally, we assume that the adversary does not have access to any random bits used by our algorithm.

%Our modified problem formulation of classic metrical task systems permits arbitrary modification of the state space. If we do not put any restrictions on the adversary, an omnipotent adversary can always remove the system's current state, forcing the system to constantly switch to new states. Thus, we adopt the classic \emph{oblivious} adversary model~\cite{borodin2005online}, meaning that an adversary must prepare the entire sequence for queries and state modifications in advance. 

\subsection{Framework Overview}
\label{sec:algoverview}
\sysname consists of two main components, the \layout and the \reorg (Figure~\ref{fig:overview}). 

The \layout is the producer of the dynamic state space. It constantly generates new data layouts according to recent samples of the query stream and issues state management queries to add and remove states to and from the state space. Examples of layout generation methods include simple heuristics such as sorting and Z-ordering~\cite{morton1966computer} with user-defined columns, as well as more sophisticated techniques that directly make use of query workloads such as bottom-up row grouping~\cite{sun2014fine} and Qd-tree~\cite{yang2020qd,ding2021instance}. The layout manager is agnostic to the underlying data layout generation mechanism as long as it supports the following two functionalities:
\begin{itemize}[leftmargin=*,topsep=1pt]
\item \texttt{generate\_layout}($\cD, \cQ, k$): This procedure uses the given dataset sample $\cD$, query workload $\cQ$, and target number of partitions $k$ to generate a data layout, which is a mapping function that assigns data records to partitions. 
\item \texttt{eval\_skipped}($s, \cQ$): This procedure estimates the fraction of data partitions skipped on the given query workload $\cQ$ and data layout $s$. 
\end{itemize}

\begin{figure}[t]
    \centering
    \includegraphics[width=0.8\linewidth]{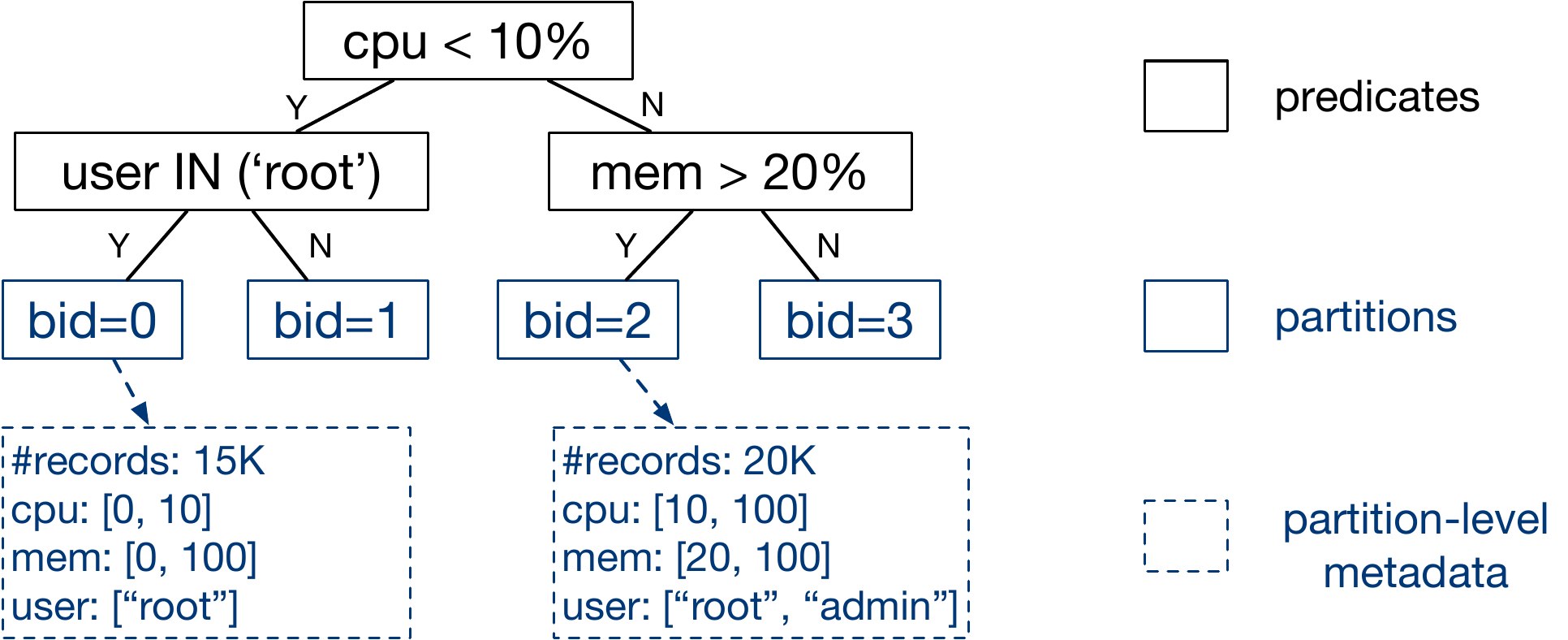}
    \caption{Example of a Qd-tree and its partition-level metadata. }
    \label{fig:qd-example}
    \vspace{-1em}
\end{figure}

As a concrete example, suppose Qd-tree is used to partition the dataset. The layout manager first calls \texttt{generate\_layout} to compute a candidate layout on a small sample of the dataset (e.g., 0.1\% to 1\% of the data). Note that creating layouts based on these small samples is considerably less resource-intensive compared to executing a full reorganization. Furthermore, previous studies have demonstrated that layouts derived from small samples are sufficiently accurate representations of those that would be generated from the complete dataset. \texttt{generate\_layout} returns the generated layout in the form of a binary decision tree, with each inner node containing a predicate selected from the query workload (Figure~\ref{fig:qd-example}). We can assign data records to partitions by routing records through this tree until they reach one of the leaf nodes. In addition, by comparing the query with partition-level metadata (e.g., partition size, range/distinct values of each column) stored in the tree, the query optimizer can identify a list of partitions that can be skipped for the query. Since \texttt{generate\_layout} operates on a small sample and \texttt{eval\_skipped} references metadata instead of the actual data, the layout manager can efficiently explore multiple layout candidates.

The second component, the \reorg, is the consumer of the dynamic state space. It queries the state space to get a set of data layouts that the system can choose from at a given time. During query processing, the \reorg observes each new query and makes the decision of whether to keep the current data layout or switch to a new layout using D-UMTS. Similar to common practices in current systems~\cite{armbrust2020delta, snow}, we assume that reorganization happens via a separate process in the background using a (partial) copy of the data and that queries are still serviced on the existing data layout while reorganization is in progress. After reorganization is completed, the new layout is swapped with the existing layout with minimal impact on the query performance.

%To avoid disruption of query performance while reorganization is in place, we would like to have reorganization happen in a separate background process using a copy of the data, without blocking readers. That is, queries are still serviced on the existing data layout until the reorganization process completes. Except during reorganization, we do not keep additional copies of the data with different layouts. This is consistent with the practices of existing systems like Delta Lake~\cite{armbrust2020delta} and Snowflake~\cite{snow}, whose data layout optimization processes operate in a non-blocking fashion in the background.

\subsection{\revision{Assumptions and Limitations}}
\label{sec:limitations}
\revision{
We consider dynamically optimizing data layouts for a static dataset. For streaming data that is ingested continuously, reorganizing the entire dataset with each new data point arrival is not practical. Instead, we could batch newly arrived data and reorganize them separately from the already ingested data. This approach resembles existing layout optimzation practices that allow users to incrementally change clustering keys by only applying the change on newly ingested data~\cite{liquid}.

In distributed settings, we assume that each node can independently reorganize its local data in response to queries directed at it. For vertically partitioned data, it makes sense to handle frequently accessed partitions separately from less accessed ones, as the former would likely benefit more from dynamic reorganization. In horizontal partitioning, dynamically changing the sharding key after the system is in operation could be challenging~\cite{bestpractice}, which motivates the setup where each shard independently manages its reorganization. However, dynamically changing the sharding key is beyond the scope of our paper, as it entails additional considerations such as load balancing that is not covered by the D-UMTS framework.

We focus on scenarios where the time taken for reorganization is relatively short compared to the rate at which query patterns change. For example, we have observed that in production query workloads at VMware's internal data platform SuperCollider, the query patterns remain stable over short periods (e.g., days) but can shift over longer spans (e.g., months). Rapidly changing query patterns could render new data layouts outdated by the time reorganization is complete. In practice, systems can also mitigate this risk by dynamically allocating resources to ensure reorganization finishes efficiently~\cite{snowauto}. 

 }
\section{Dynamic Reorganization via MTS}
\label{sec:reorg}
The first component of \sysname is the MTS-based \reorg. The \reorg observes the query stream and decides when the system should switch to a different data layout and if so, which layout to switch to. On a high level, the \reorg extends the classic algorithm of Borodin, Linial, and Saks~\cite{BLS92} to allow arbitrary modification of the state space during query processing.  

\subsection{Overview of the Algorithm of Borodin, Linial, and Saks}
We first present an overview of the classic algorithm~\cite{BLS92} (Algorithms~\ref{alg:process-queries}-\ref{alg:update-counters}). The algorithm takes as input a query stream $\cQ$ and a fixed set of states $\cS$, and returns a set of states $\cH$, one per query, that the system should be in \emph{before} processing the query. 

\begin{algorithm}[ht!]
\footnotesize
\caption{$\textbf{ProcessQueries}(\cQ, \cS)$ \label{alg:process-queries}}
\begin{algorithmic}[1]
\Statex{Processes queries in ordered set $\cQ$ with states in $\cS$.}
\Statex{$\cC$: counters, $\cS_A$: active states}
\State{$\cS_A, \cC \gets \textrm{ResetStates}(\cS)$} 
\State{$s_c \sim \textrm{Uniform}(\cS_A)$} \Comment{Randomly select a current state}
\State{$\cH \gets \{s_c\}$} 
\For{$q \in \cQ$}
    \State{$s_c, \cS_A, \cS, \cC \gets \textrm{UpdateCounters}(q, s_c, \cS_A, \cS, \cC)$}
    \State{$\cH \gets \cH \cup \{s_c\}$}
\EndFor
\State \Return $\cH$
\end{algorithmic}
\end{algorithm}

\begin{algorithm}[ht!]
\footnotesize
\caption{$\textbf{ResetStates}(\cS)$ \label{alg:reset-states}}
\begin{algorithmic}[1]
\Statex{Creates new active states $\cS_{A}$ and resets all counters.}
\State{$\cS_A \gets \cS$}
\State{$\cC(s) \gets 0$ for $s \in \cS_A$}
\State \Return $\cS_A, \cC$
\end{algorithmic}
\end{algorithm}

\begin{algorithm}[ht!]
\footnotesize
\caption{$\textbf{UpdateCounters}(q, s_c, \cS_A, \cS, \cC)$ \label{alg:update-counters}}
\begin{algorithmic}[1]
\Statex{Updates the counters of states after processing query $q$ with current state $s_c$.}

\State{$\cC(s) \gets \cC(s) + c(s, q)$ for $s \in \cS_A$ \label{algline:counter-increase}} \Comment{Update counters}
\State{$S_A \gets \{s\,\mid\,\cC(s) < \alpha \textrm{ for }s \in \cS_A\}$}  
\If{$s_c \not\in \cS_A$} \Comment{If current state counter is full}
    \If{$\cS_A = \emptyset$} \Comment{All counters are full}
        \State{$\cS_A, \cC \gets \textrm{ResetStates}(\cS)$} \Comment{Start new phase}
    \EndIf
    \State{$s_c \gets \mathrm{Uniform}(\cS_A)$} \Comment{Switch to new state} 
\EndIf
\State \Return $s_c$, $\cS_A, \cS, \cC$
\end{algorithmic}
\end{algorithm}

\Cref{alg:process-queries} initializes a set of ``counters" starting at 0 for each state in $\cS$. For each new query $q$ that is processed, each counter increases by the cost $c(s, q)$ of serving $q$ in state $s$. A counter is considered ``full" when the cumulative cost is at least $\alpha$. When the counter for the current state is full, the system randomly switches to another state whose counter is not full with uniform probability. When there are no more states to switch to (all counters are full), the algorithm ``resets" all counters back down to 0 (\Cref{alg:reset-states}). The periods between the state resets are called ``phases" of the algorithm. As shown in~\cite{BLS92}, the algorithm has a competitive ratio of $O(\log|\cS|)$.

Intuitively, the counters measure the cost a state would have incurred if it processed all the queries in a phase. All the counters are at least $\alpha$ after a phase, so the optimal algorithm is guaranteed to incur a cost of at least $\alpha$ (had it stayed in a state during the entire phase). This is the key to bounding the performance of the optimal algorithm during the analysis.

Applying this algorithm to our layout optimization problem has the following implications. The states are different data layouts, and the service cost is the fraction of data accessed from running the query on the current data layout, which can be estimated using a small amount of partition-level metadata. Given that the relative overhead of reorganization versus querying is $\alpha$, each phase ends when all counters are above $\alpha$. Note that the algorithm's behavior is influenced by $\alpha$ -- higher values of $\alpha$ generally result in fewer reconfiguration moves, while lower values of $\alpha$ will result in more reconfiguration moves. The effect of $\alpha$ is evaluated in Section~\ref{sec:evalparams}.

In the original algorithm, each phase begins with transitioning to a random state. This makes the analysis of the algorithm simpler and cleaner. To reduce reorganization costs, a simple but effective optimization is simply allowing the algorithm to stay in the current state when starting a new phase instead of forcing it to move to a random state, thus saving on this initial random transition cost. 
%The intuition is that, if the current state is performing well, we do not want to risk moving towards a random state that is likely to have worse performance. 
Since each phase is independent, this optimization does not asymptotically affect the competitive ratio of the randomized algorithm. Empirically, we have observed that having the option to stay in the current state significantly improves the reorganization cost. 

However, one difficulty lies in how to initialize the state space $\cS$. Although the algorithm assumes that we can work with a fixed set of states throughout query processing, in practice we do not know this set apriori. Without knowledge of the query workload, what we would like to do is to start with a default data layout, such as Z-ordering on one or more predefined sort columns, and as the systems observe more queries, compute better-performing, workload-aware layouts to add to the state space $\cS$ incrementally. The next section describes how we have modified this classic algorithm to permit such modification of the state space $\cS$ during query processing.

\subsection{Supporting Dynamic State Spaces}
\label{sec:dss}
We model the modification of the state space via state update queries, which add and remove any state in the state space during the processing of service queries. The flexibility of removing states allows for controlling the size of the state space in our applications (e.g., to bound the amount of metadata maintained by the algorithm). Algorithm~\ref{alg:update-states} summarizes our proposed modification to handle state update queries in the framework. If a new state is added in the middle of the phase, we simply defer the new state to the next phase. In other words, the algorithm behaves as if no additions have happened in the current phase, and resets the active state space to include any new states when the next phase starts. If an existing state is deleted in the middle of a phase, we set the counter of this state to $\alpha$ to mark that the algorithm can no longer switch to this state in the current phase. If all remaining counters are full after removing this state, we simply reset all counters to 0 and start a new phase with the updated state set. If the state that our system is currently in gets deleted, we follow the same procedure as if the current state is full, and randomly switch to another state that is still available.

Since the original guarantee of an $O(\log|\cS|)$ competitive ratio is no longer applicable in the dynamic setting, one open question is how well do our proposed modifications work for \dumts. The following analysis shows that our algorithm costs no more than $O(\log |\cS_{max}|)$ in a phase, where $|\cS_{max}|$ is the maximum size of the active state space over the course of the query stream. Moreover, this competitive ratio is asymptotically optimal. In particular, we show the following:

\begin{theorem}
\label{thm:apx-guarantee}
\Cref{alg:update-states} solves \dumts with competitive ratio $2H(|S_{max}|) \leq 2(1+\log |\cS_{max}|)$ where $H(n)$ is the $n$-th harmonic number and $\cS_{max}$ is the largest set of states created by the update queries over the course of the stream.
\end{theorem}

\begin{proof}
As discussed above, we break the execution of the algorithm into ``phases", which are the execution intervals of the algorithm between calls to ResetStates (\Cref{alg:reset-states}). In other words, a phase begins when all counters in $\cS_A$ are set to 0, and ends when $\cS_A$ is empty. During a phase, the total set of states $\cS$ may change as states are added or deleted. The counters have a natural interpretation: they are the lower bound of the cost a state would have incurred had an algorithm chosen to stay in that state during the entirety of the phase.

Let $\cS_f$ and $\cS_i$ be the set of states in $\cS$ at the beginning and end of a phase. Let $\cA_{opt}$ denote the optimal offline algorithm, and consider its cost in a phase. There are two cases, either $\cA_{opt}$ moves to a newly added state in $\cS_f \setminus \cS_i$, or it remains in one of the original states in $\cS_i$. In the former case, $\cA_{opt}$ incurs movement cost at least $\alpha$. In the latter case, every state in $\cS_i$ has a counter at least $\alpha$, which implies that $\cA_{opt}$ incurs total cost at least $\alpha$.

On the other hand, we show that the cost of our algorithm is bounded by $2H(|S_i|)$. Let $f(k)$ be the cost of our algorithm when $|\cS_A| = k$ and note that $f(0) = 0$. We say a counter for state $s$ has ``reached" $\alpha$ if $\cC(s)$ exceeds $\alpha$ at the end of line~\ref{algline:counter-increase} in \Cref{alg:update-counters}. Let $s_i$ be the $i$-th state in $\cS_A$ whose counter reaches $\alpha$, breaking ties arbitrarily. Note that being in state $s_i$ means at most $|\cS_A|-i$ states remain when we transition out of $s_i$. With equal probability, we transition to any state in $\cS_A$, so $$f(k) \leq 2\alpha + \sum_{i=1}^k f(k-i)/k = 2\alpha/k + f(k-1) = 2 \alpha H(k),$$ where the factor of $2$ accounts for the cost incurred by our current state serving queries (filling its counter from 0 to $\alpha$) and from transitioning to a new state (movement cost $\alpha$). The theorem follows by applying the above argument to every phase of the algorithm and noting that $H(|S_i|) \leq \log (|S_i|) + 1$.
\end{proof}

\begin{remark}
\Cref{alg:update-states} is also asymptotically optimal in the competitive ratio, as the competitive ratio for the non-dynamic uniform metrical task system is lower bounded by $\log(|S_{max}|)$~\cite{BLS92}.
\end{remark}

\begin{algorithm}[t]
\footnotesize
\caption{$\textbf{ProcessQueries*}(\cQ, \cS)$ \label{alg:update-states}}
\begin{algorithmic}[1]
\Statex{Processes queries in ordered set $\cQ$ with states in $\cS$.}
\State{$\cS_A, \cC \gets \textrm{ResetStates}(\cS)$} 
\State{$s_c \sim \textrm{Uniform}(\cS_A)$} \Comment{$s_c$: current state}
\State{$\cH \gets \{s_c\}$} 
\For{$q \in \cQ$}
    \If{$q$ removes state $s_d$} \Comment{Remove state}
        \State{$\cS_A \gets \cS_A \setminus \{s_d\}$}
        \State{$\cC(s_d) \gets \alpha$}
        \If{$\cS_A = \emptyset$} \Comment{No state remains}
            \State{$\cS_A, \cC \gets \textrm{ResetStates}(\cS)$}
        \EndIf
        \If{$s_d = s_c$} \Comment{If current state is deleted}
            \State{$s_c \sim \textrm{Uniform}(\cS_A)$} \Comment{Switch to a random state}
        \EndIf
    \ElsIf{$q$ adds state $s_a$} \Comment{Add state}
        \State{$\cS \gets \cS \cup \{s_a\}$}
    \Else \Comment{Service query}
    \State{$s_c, \cS_A, \cS, \cC \gets \textrm{UpdateCounters}(q, s_c, \cS_A, \cS, \cC)$}
    \State{$\cH \gets \cH \cup \{s_c\}$}
    \EndIf
\EndFor
\State \Return $\cH$
\end{algorithmic}
\end{algorithm}

\subsection{Improving Reorganization with a Predictor}
\label{sec:learned}
So far, our proposed algorithm makes no distinction between states and switches between states whose counters are not full randomly with uniform probability. In practice, users may have additional information on which states have better performance. In this section, we describe how to leverage such additional information to improve the performance of the algorithm. 

We assume there exists a predictor $p(s, \cS_A)$ which outputs transition scores for choosing the next state $s \in \cS_A$. A higher transition score indicates a state that is predicted to be more efficient for the upcoming workload. A user can then use this predictor to construct a transition distribution where for each active state $s \in \cS_A$, we jump to a state $s$ with probability $\frac{p(s, \cS_A)}{\sum_{s \in \cS_A} p(s, \cS_A)}$. In our algorithm, we would like to jump to the state in $\cS_A$ that is the most efficient in its phase. Rank the states of $S_A$ by their efficiency (each state $s$ gets a rank from $\{1,\ldots, |S_A|\}$, with 1 being the most efficient). When we choose uniformly at random, we get a state with $|S_A| / 2$ on average. The hope is that the predictor $p$ biases our jump towards the states with ranks closer to 1. More formally, the predictor $p$ helps us construct a discrete distribution $\cD_{|\cS_A|}$ on $\left\{1,\ldots, |\cS_A|\right\}$, where making a jump corresponds to sampling from $\cD_{|\cS_A|}$.

Concretely, in our application,  we can use the query costs incurred by states in the previous phase as the predictor for their performances in the next phase. For states that are added in the middle of a phase, we can replay the queries processed in the current phase so far to fill in the counter, or simply initialize the counter to be the median of query costs incurred so far by existing states in the phase. When the counter for the current state is full, the algorithm can pick a new state based on a distribution that favors the better-performing states, instead of choosing a new state to switch to uniformly at random. For example, we can assign each state $s$ with a weight $w_s$ proportional to the average fraction of data skipped in the last phase, and switch to a new state $s$ with probability $\frac{w_s^\gamma}{\sum_{s\in \cS_A} w_s^\gamma}$. For $\gamma=0$, this is equivalent to choosing the next state under a uniform distribution. For $\gamma>0$, the probability distribution favors states that have larger weights, or better performances in the last phase. In fact, we can show that the competitive ratio is directly influenced by the accuracy of the distribution in predicting the ``optimal state" over a single phase.

\begin{theorem}
\label{thm:mladvice}
Fix any phase of \Cref{alg:process-queries}, and let $\cS$ be the initial states at the beginning of the phase. Order the states of $\cS = \{s_1, s_2, \ldots, s_n\}$ in order of the time their counters reach $\alpha$. Let $\cD_m$ be distributions $\{1, 2, \ldots, m\}$ such that $\E[\cD_m] / m \geq \beta$ for all $m=1,2,\ldots,n$.\footnote{I.e., our predictor is expected to give something within the top $\beta$ fraction of ranks.} Then the expected competitive ratio of \Cref{alg:process-queries} is at most $O\left(\log_{(1-\beta)^{-1}} n\right)$.
\end{theorem}
\begin{proof}
Recall that in the analysis of \Cref{alg:process-queries}, random transitions are made whenever a counter reaches $\alpha$, and that the competitive ratio is proportional to the number of transitions made. Let $\cS_i$ be the active states left on the $i$-th transition. To bound the number of transitions, let $X_i$ be the indicator random variable for the event that $|\cS_{i+1}|\leq 2(1-\beta) |\cS_i|$, where transition $i$ is made according to distribution $\cD_i$. By Markov's inequality, $\textrm{Pr}[X_i] \geq 1/2$. To reduce $n$ states down to 0, we need $O\left(\log_{(1-\beta)^{-1}} n\right)$ of the indicators to be 1. Since each indicator has a positive constant probability of occurring, this implies that we get $O\left(\log_{(1-\beta)^{-1}} n\right)$ transitions in expectation.
\end{proof}
In other words, the competitive ratio of our algorithm improves if a predictor gives distributions that are biased towards the most efficient states. We evaluate the impact of the transition distribution on the overall performance of the algorithm in Section~\ref{sec:evalparams}.
\section{On-the-fly Layout Generation}
\label{sec:layout}
The second component of \sysname is the \layout. As discussed earlier, without knowledge of the query workload, 
the system needs to generate new data layouts incrementally as it observes more queries.
%we cannot precompute a small set of relevant data layouts for the system to choose from. Instead, the system needs to generate new data layouts incrementally as it observes more queries. 
In this section, we present the design of the \layout that determines 1) which new states to compute and 2) whether to admit these new states into the dynamic state space.

\subsection{Generating Diverse Data Layouts}
\label{sec:sw}
Workload-aware data layouts can experience performance degradation under changing query workloads. Therefore, it is necessary to periodically update data layouts to keep up with the changes. There are a few common strategies for performing this update. 

The first strategy is to periodically compute new data layouts based on all queries that have arrived so far. However, this strategy is expensive as the query history keeps growing over time. It is also ineffective since recent queries can get out weighted by the larger volumes of historical queries. The second strategy is to keep a sliding window (of a fixed size or of a fixed time interval) of recent queries and periodically compute new layouts based on queries in the sliding window. The sliding window is bounded in size and responsive to changes, but has no memory of the distant past. The third strategy is to keep a reservoir sample of queries, which biases towards recent events but also keeps memories from the past. Reservoir samples provide a holistic picture of the entire history with a limited memory budget. However, since reservoir sampling always keeps around a small portion of old queries, layouts generated from the reservoir sample tend to perform slightly worse on the current workload compared to layouts generated from the sliding window. 

We experimented with the use of both sliding windows and reservoir sampling for generating data layout candidates. We empirically found that using layouts generated solely from sliding windows gives the best overall performance (\S~\ref{sec:evalparams}). One explanation is that, since the cost of switching between layouts is constant, it is more beneficial for the online algorithm to switch between layouts that have good performance for a specific workload, rather than layouts that have mediocre performance for multiple workloads. For example, consider a workload that iterates through each column of the dataset and generates 100 random range queries per column. Since reservoir sampling incorporates historical data, the sampled workload would always include queries on multiple columns. As a result, the \layout is unable to generate the "optimal" layouts that partition according to a single column.

%As a concrete example, consider the synthetic workload described in \S~\ref{sec:synexample}. Since reservoir sampling incorporates historical data, the sampled workload would always include queries on multiple columns. As a result, the \layout is unable to generate the "optimal" layouts that partition according to a single column.

%Our layout manager combines the responsiveness of the sliding window with the memory of the reservoir sample. Concretely, we keep several independent reservoir sample pools with different decay rates in addition to a sliding window so that we can get query samples with different mixes of recent and history events. Each query sample therefore allows the offline module to generate a different optimal layout. The next subsection describes how the layout manager expands the dynamic state space based on these new data layouts. 

\begin{algorithm}[t]
\footnotesize
\caption{$\textbf{Layout Management}$}
\label{alg:layout-candidates}
\begin{algorithmic}[1]
\Statex{$\cS_A$: dynamic state space, $s$: new state}
\Statex{$Q$: sample of past queries, $\epsilon$: distance threshold}
\Procedure{admit\_state}{$\cS_A, s, Q$}
\State{$c =$ \textsc{eval\_skipped}($s, Q$)}
\State{$dists \gets []$}
  \For {$s_i \in \cS_A$}
  \State{$c_i =$ \textsc{eval\_skipped}($s_i, Q$)}
  \State{$dists.\text{add}\left(\frac{||c- c_i||_1}{\text{dim(c)}}\right)$}
  \EndFor
 \State \Return $min(dists) > \epsilon$
\EndProcedure
\end{algorithmic}
\end{algorithm}

\subsection{Expanding the Dynamic State Space}
Now that new data layouts are constantly being generated, the \layout needs to decide how to update the state space accordingly. One option would be to simply admit all newly generated states. However the competitive ratio of the algorithm is directly related to the size of the dynamic state space. On the other hand, if we do not expand the state space at all, the system might miss out on new data layouts that are better performing for current and future workloads. Therefore, the \layout needs a policy that can determine whether to admit a new data layout into the dynamic state space. 

We draw inspiration from prior work in online learning that tries to improve performance guarantees by reducing the size of state space~\cite{cohen2017online}. For example, one can find a small number of states that form a covering of the space such that at least one of them is close to the optimal state for the system. Using similar reasoning, admitting highly similar data layouts to the dynamic state space does not help much with the query performance. Furthermore, it can incur additional reorganization costs as the reorganization might end up switching back and forth between similar layouts. 

In our setup, we consider two data layouts to be similar if they incur similar query costs over the query stream. Specifically, we use a reservoir-based time-biased sampling (R-TBS) algorithm proposed in~\cite{hentschel2019general} to curate a representative query sample of size $s$ over the query stream. We evaluate each data layout candidate $i$ on the sample to get a cost vector $c_i = (c_{i,1}, ..., c_{i,n})$, where $c_{i,n}$ is the cost of executing query sample $n$ on data layout $i$. We define the difference between two data layouts $i$ and $j$ to be a distance function (e.g., normalized L1 distance) of the query cost vectors $c_i$ and $c_j$. The \layout is configured with a distance threshold $\epsilon \in [0, 1]$ and only admits a new data layout if it is at least $\epsilon$ distance apart from all existing layouts in the state space. As the size of the state space becomes larger, it is increasingly difficult for a new data layout to get admitted, since it needs to be different enough from {\em all} existing states in the space. The pseudo-code for the procedure is presented in Algorithm~\ref{alg:layout-candidates}.

Since the query samples are updated as the system processes more queries, the similarity between data layouts measured on these query samples can also change over time. For instance, two data layouts that initially had a distance of $> \epsilon$ when admitted to the state space may have a distance of $< \epsilon$ under query samples at a later time. This is acceptable because at the time when the \layout decides to admit the new state, switching to this state would make a difference in query performance. The \layout can also periodically prune the state space based on recent query samples and remove candidates that incur similar query costs to other layouts. We investigate the impact of the distance threshold $\epsilon$ on the size of the dynamic state space and the performance of our algorithm in Section~\ref{sec:evalparams}. We also compare the dynamic state space to a fixed state space that is precomputed based on knowledge of the query workload in Section~\ref{sec:oracle}.

\section{Evaluation}
\label{sec:eval}

In this section, we empirically evaluate \sysname on a variety of datasets and workloads. Overall, the experiments show that 
\begin{itemize}
    \item Compared to using a single data layout, dynamically switching layouts using \sysname results in up to 32\% improvement in end-to-end runtime on various datasets and workloads.
    \item \sysname remains competitive and often outperforms alternative online reorganization strategies. 
    \item \sysname works across different data layout generation mechanisms and is robust to changes in key parameters.
\end{itemize}

\subsection{Setup}
\label{sec:eval-setup}
\subsubsection{Implementation} \sysname is implemented as a lightweight Python library. \sysname keeps track of different data layouts via partition-level metadata. With information such as row count, range of values (or distinct values for categorical columns) for each column in the partition, \sysname is able to estimate query costs incurred by different layouts without accessing the underlying dataset. 

We experimented with two data layout optimization techniques in the evaluation: Z-ordering and Qd-tree. First, we use Z-ordering on user-defined columns to split the dataset into equal-sized partitions. To make Z-ordering workload-aware, we use the top three most queried columns in the sliding window, which can change over the course of the query stream. Second, we construct data layouts using Qd-trees, a workload-aware layout optimization technique. Our implementation of the Qd-tree uses the greedy construction algorithm and does not include any advanced cuts. Similar to practices in prior work, we construct Qd-trees based on a 0.1\% to 1\% sample of the dataset. 

We perform evaluations on a VM with 64GB of RAM and 8 Intel(R) Xeon(R) Gold 6230R 2.10GHz CPUs. For methods that use variations of the classic MTS algorithm, we report the average from three runs. We report results using two metrics: physical runtimes in end-to-end experiments and logical costs in simulation.

\minihead{End-to-end}
We report end-to-end query execution time and reorganization time from a shallow integration with Apache Spark, \revision{similar to setups in prior works~\cite{ding2021instance,sudhir2023pando}. } Specifically, we include a new column for each table that specifies the unique partition ID (BID) that each row maps to. During query processing, we use the partition-level metadata to calculate a list of BIDs that the query needs to read. We then rewrite the original query to include an additional predicate that filters by the list of computed BIDs. For example, if the metadata indicates that only partitions 6 and 10 are relevant for the query, we write the query predicate to include an explicit partition filter \texttt{BID IN (6, 10)}.   During reorganization, we update the BID column according to mappings generated from Qd-tree/Z-ordering. We rewrite rows with the same BID into a new partition, which is stored as a Parquet file on local disks. We exclude the time taken to compute data layouts since it is dependent on the specific layout optimization technique and is not the focus of this work. In general, data layouts can be computed rather efficiently on large datasets~\cite{armbrust2020delta,yang2020qd}. We run Spark in stand-alone mode estimate the total query time using a sample of 2000 queries (around 10\% of the workload). 

\minihead{Simulation} We also report results using logical costs in simulation. Per our cost model, we use the fraction of data records accessed (between 0 and 1) as a proxy for the query costs and assume the reorganization cost to be $\alpha$. We investigate the effect of varying key parameters on the overall performance of the framework as well as the impact of our proposed optimizations in simulation.

\begin{figure*}[ht]
    \centering
    \includegraphics[width=0.9\linewidth]{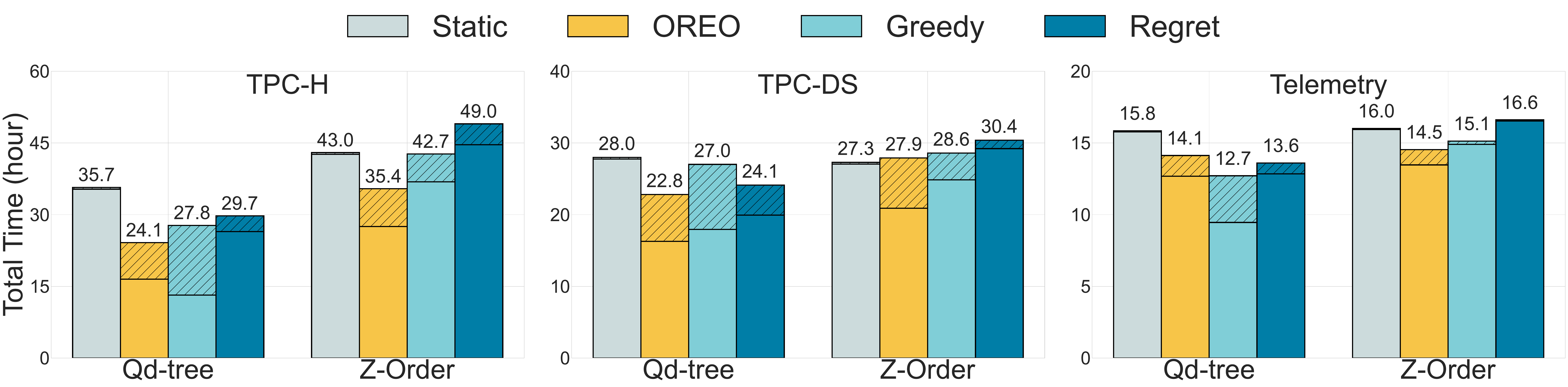}
    \caption{Comparison of total query and reorganization time in Spark enabled by \sysname with baselines. The top half of each bar with the hatches represents reorganization time, while the bottom half represents query time. Overall, dynamic reorganization improves upon a static, optimized layout by up to 32\% in total compute time.  }
    \label{fig:spark}
    \vspace{-1em}
\end{figure*} 

\subsubsection{Dataset and Query Workloads} We perform evaluation on three real-world datasets: TPC-H, TPC-DS and a production workload from an internal data lake system at VMware (Telemetry). We set the target partition count such that each partition contains between 1 million to 2 millions rows, and the average size of the partitions (Parquet files) is between 100 to 200 megabytes, consistent with recommended best practices. 

\minihead{TPC-H} We use the TPC-H data generator with a scale factor of 100 and denormalize all tables against the lineitems table. Due to the uniform nature of TPC-H data, we pre-divide the dataset into 10 equally sized partitions according to its primary keys and perform reorganization on one of the partitions, which contains around 40 millions rows, 58 columns. Similar to prior work~\cite{yang2020qd,sun2014fine}, we include 13 TPC-H query templates that touch the lineitem fact table (q1, q3, q4, q5, q6, q7, q8, q10, q12, q14, q17, q21\footnote{q9 and q18 are excluded because they involve predicates that can not be directly evaluated using basic partition-level metadata. Specifically, q9 includes \textsc{LIKE} operator on a high cardinality column \textsc{P\_NAME} and q18 includes a filter on the aggregate value of a group by.}). The workload generator behaves like a state machine and samples queries from one query template for an arbitrary amount of time before switching to another random query template. The workload contains a total of 30,000 queries, generated from 20 query templates. %and with an average selectivity of \red{TODO}.

\minihead{TPC-DS} We use the TPC-DS data generator with a scale factor of 10 and denormalize all tables against the \texttt{store\_sales} table. The resulting table has around 26 million rows. We include 17 TPC-DS query templates that involve queries on the \texttt{store\_sales} fact table and the dimension tables~\footnote{We used q3, q7, q13, q19, q27, q28, q34, q36, q46, q48, q53, q68, q79, q88, q89, q96, q98.}. The workload includes 30,000 total queries, generated from 20 query templates using the same method as the TPC-H workload. %The average selectivity of the workload is \red{TODO}.

\minihead{Telemetry} \anonsys is an internal data platform that is used by over a hundred teams for data analytics in VMware. \anonsys allows users to set up data jobs for bulk ingestion, run SQL queries over the data, discover and share datasets in a market space. We investigate a production table in \anonsys that logs monitoring information for ingestion jobs. We extract a sample of the table as well as queries that touch this table in the past six months. In total, our samples include 24,000 queries and around 30 million rows for the table. The most popular predicates include range queries on the arrival time of the record, where the time interval ranges from a few hours to a few months, as well as filters on the name of the collector who has sent the data. %The overall workload selectivity is \red{TODO}.

\subsubsection{Methods of Comparison} \label{subsec:comp_methods}
We compare \sysname against one offline baseline (Static) which observes the entire query workload in advance but is not allowed to change states, as well as two online baselines (Greedy and Regret) that do not have workload knowledge but can change states. \revision{The three online approaches (Greedy, Regret and \sysname) utilize the same set of data layout candidates computed periodically based on a sliding window of recent queries, but use different reorganization strategies.} By default, the relative reorganization cost $\alpha$ is set to 80 based on measurements obtained on our system setup, and the sliding window is set to include the most recent 200 queries.

\minihead{Static} The method observes the entire query workload in advance and constructs a single layout that optimizes data skipping for the entire workload. 

\minihead{\revision{Greedy}} The method compares the performance of the current data layout with a new data layout computed based on a sliding window of recent queries, \revision{ and greedily switches to the new layout if it has a smaller query cost than the current one, without considering the reorganzation cost.}

\minihead{Regret} \revision{This method is similar to the Greedy strategy but considers the reorganization cost,} inspired by work on storage management in video analytics~\cite{daum2021tasm}. The method keeps track of the cumulative difference in query costs between the current data layout and alternative layouts over the query history. For each new layout, the method retroactively computes performance improvement compared to the current layout, using all queries that have been serviced on the current layout. The method switches to a new layout when the cumulative saving in query cost exceeds the reorganization cost. 

\minihead{\sysname} A prototype that matches the descriptions given so far. Unless otherwise specified, the default parameter values for \sysname in the experiments are $\epsilon=0.08$ and $\gamma=1$.

\subsection{End-to-end Results}
\label{sec:evalmain}
Figure~\ref{fig:spark} summarizes the end-to-end query and reorganization time incurred by different methods. The bottom half of each bar represents query time, while the top half with the hatches represents reorganization time. On the datasets that we experimented with, dynamic reorganization can outperform a single, precomputed data layout for the entire workload (Static) by 32.5\%, 18.6\% and 10.8\% respectively using Qd-trees as the underlying partitioning technique. Compared to using Qd-trees, layouts generated using simple heuristics such as Z-ordering tend to have worse data skipping ratio, as reflected by the increased query costs. One notable exception is that for the TPC-DS dataset, the static Z-ordering layout outperforms the Qd-tree layout despite having a worse data skipping ratio. This is because the Z-ordering layout has a better compression ratio (16\% smaller in file size) and the compression ratio is not taken into account in the design of Qd-trees. On the other two datasets, \sysname is still 17.6\% and 9.4\% better compared to the best static layout. 

Among the three online reorganization strategies, \revision{the Greedy baseline is the most aggressive in reorganization: it always switches to a better-performing layout regardless of the reorganization cost. Therefore, this baseline represents the smallest query costs that can be obtained by online strategies that share the same set of layout candidates. However, the downside is that it incurs large reorganization costs, especially when $\alpha$ is large. } 
%This is reflected in this baseline's overall small query costs and large reorganization costs. 
In contrast, the Regret baseline is most conservative, thus having overall large query costs and small reorganization costs. \sysname lies somewhere in between and achieves the best overall cost in all but one case. In addition, Greedy and Regret both become less effective with Z-ordering. Since layouts generated from Z-ordering generally offer smaller improvements in query performance, the two baselines end up making fewer layout changes and incurring query costs that are close to those of the static layout. In comparison, \sysname remains dynamic and offers noticeable improvements in query costs compared to the best static layout. 

\begin{figure}[t]
\centering
    \includegraphics[width=\linewidth]{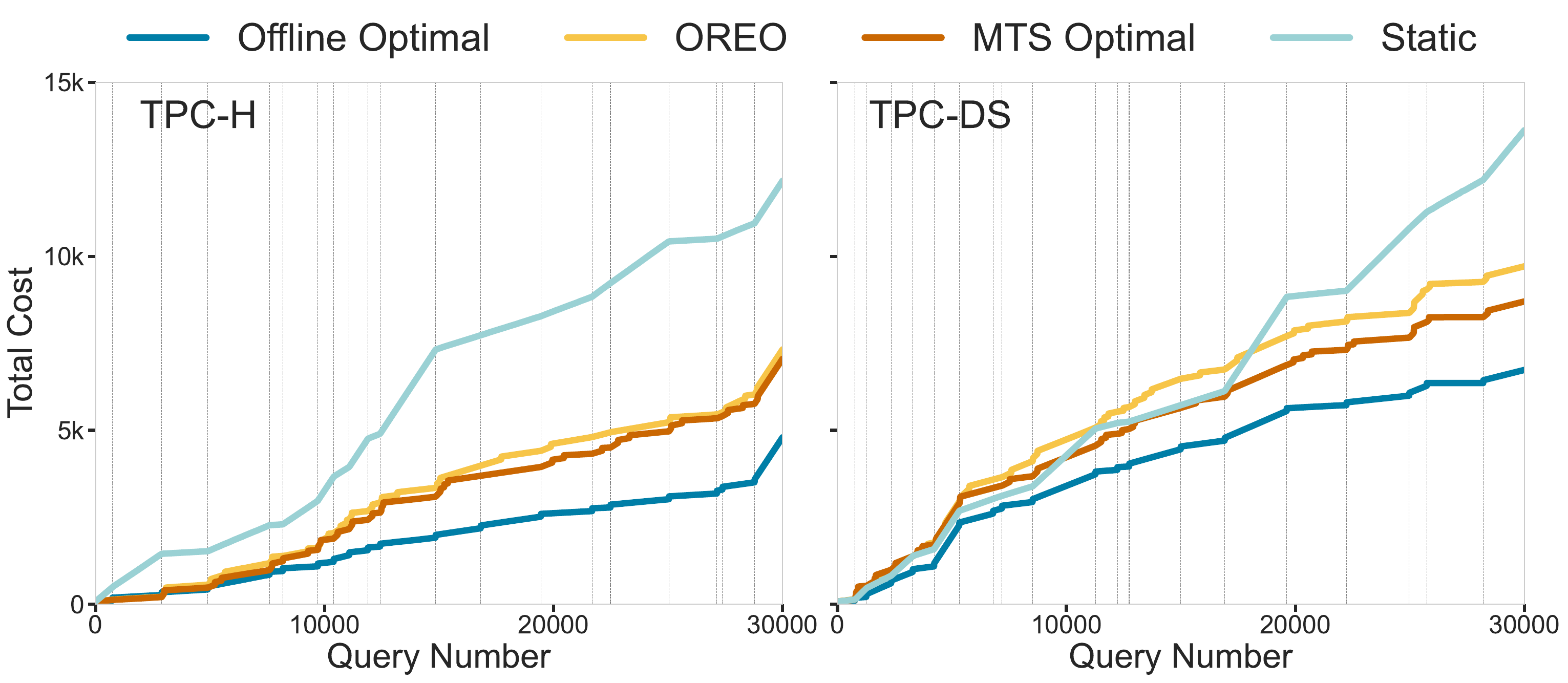}  
    \caption{\sysname's query costs is 74\% and 44\% larger than those of the optimal offline algorithm that observes the entire workload in advance and can switch states. Vertical gray lines indicate when the workload switches to a different query template.}
    %Visualization of total query and reorganization cost over the query stream.   }
    \label{fig:spark-line}
    \vspace{-1em}
\end{figure}

\begin{figure*}[ht]
    \centering
    \includegraphics[width=0.85\linewidth]{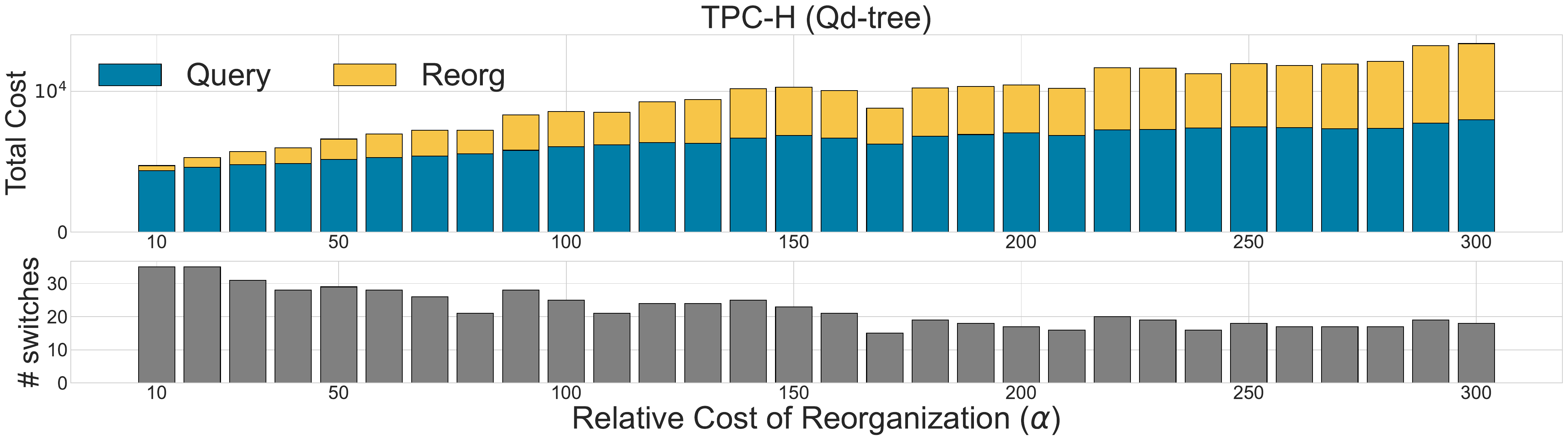}
    \caption{Impact of reorganization cost ($\alpha$) on the overall performance. As reorganization becomes more expensive, the total gains from dynamic reorganization decrease. The decrease is not monotonic due to changes in reorganization strategies. }
    \label{fig:alpha}
\end{figure*}

\subsection{Gap to Optimal Algorithms}
\label{sec:oracle}
In the previous experiments, we compared with a static offline algorithm that observes the entire workload in advance but is not allowed to switch states. Note that our theoretical bounds are with respect to {\em any} algorithm that sees the entire workload in advance, even those that are able to switch states. To further understand the quality of the solutions produced by \sysname, we compare with two additional methods that are allowed to switch states and leverage additional workload information:

\begin{itemize}[leftmargin=*,topsep=2pt]
    \item MTS Optimal: Instead of incrementally generating new data layouts, the method is given a fixed state space that includes the best data layout precomputed for each query template. The method uses \sysname's modified MTS algorithm to determine how to switch between layout candidates. 
    \item Offline Optimal: The method is presented with the entire workload in advance and switches to the best data layout for a query template as soon as template changes in the workload. Since this benchmark algorithm is given knowledge of the entire workload and has complete flexibility in changing the layout, it gives the lower bound on the query costs for any online solutions. 
\end{itemize}

Figure~\ref{fig:spark-line} visualizes how the total costs change over the course of the query stream for all methods of comparison. The vertical gray lines in the background indicate when the workload switches to a different query template. Overall, \sysname (yellow) using the dynamic state space performs slightly worse than the MTS Optimal (orange) which uses a fixed, precomputed state space. \sysname's query costs are within 14\% and 17\% of the query costs of the MTS Optimal. This shows that additional workload information can benefit online algorithms by improving the quality of the state space. In addition, \sysname's costs are 74\% and 44\% larger on the two datasets compared to the query costs of Offline Optimal (dark blue). Note that this is much better than the worst case $O(\log k)$ bound provided by the analysis, and the larger gap is because the Offline Optimal knows the entire workload and therefore does not experience any delay between template change and layout switches. MTS Optimal and \sysname on the other hand, use online algorithms and do not have access to such information. The Offline Optimal makes 20 layout changes in total, corresponding to the number of template switches. \sysname makes 22 and 29 layout changes while MTS Optimal makes 27 and 30 changes on the two datasets.

\subsection{Detailed Analysis}
\label{sec:evalparams}

In this section, we evaluate \sysname's behavior under changes in key framework parameters.

 % \minihead{Effect of reorganization cost $\alpha$} 
 \subsubsection{Effect of reorganization cost $\alpha$} 
 Figure~\ref{fig:alpha} shows the effect of relative reorganization cost $\alpha$ on the overall performance of the framework. Overall, the total gains from dynamic reorganization decrease as reorganization becomes more expensive. When $\alpha$ gets larger, the algorithm initially sticks to similar reorganization strategies despite the slightly increased reorganization cost. Eventually the reorganization becomes expensive enough that the algorithm adapts its strategy to make fewer layout changes in compensation. As a result, the number of layout changes decreases as $\alpha$ increases (35 changes at $\alpha=10$ and 18 changes at $\alpha=300$), with a few noticeable drops happening at around $\alpha$ = 80 and 170. This also explains why the overall costs do not increase monotonically as $\alpha$ increases. 
 
 In practice, users can measure typical values of $\alpha$ based on their system configuration to provide as inputs to \sysname. We run Spark in stand-alone mode with files stored using the Parquet format on a local hard drive and measure the time taken to run a full table scan SQL query versus data reorganization for files ranging from 16MB to 4GB. In particular, reorganization includes 1) reading partitions from disk 2) updating the \texttt{BID} column and 3) repartitioning the dataset based on \texttt{BID} and 4) compressing and writing the new partitions to disk. Overall, we observe that the cost ratios range from 60$\times$ to 100$\times$ (Table~\ref{tab:alpha}).

% \begin{table}[t]
% \footnotesize
%     \caption{Relative cost of reorganization over query ($\alpha$) ranges from 60$\times$ to 100$\times$ in our setup.  }
%       \centering
%         \begin{tabular}{r r r r r r}
%         \toprule
%          & \multicolumn{5}{c}{\textbf{File Size (MB)}} \\
%           \cmidrule{2-6}
%            & 16  & 64  & 256  & 1024  & 4096  \\
%           \midrule
%           Query (sec)&0.36$\pm$0.05  & 0.89$\pm$0.03  & 2.9$\pm$0.1  & 12.5$\pm$0.2  & 81.0$\pm$2.0 \\
%           Reorg (sec)&24.6$\pm$0.7  & 70.0$\pm$4.3  & 276.6$\pm$4.1 & 1231.3$\pm$43.9 & 4854.1$\pm$88.9\\
%           $\alpha$&69.0 & 78.7 & 95.4 & 98.4 & 59.9 \\
       
%         \bottomrule
%     \end{tabular}
%     \label{tab:alpha}
% \end{table}

\begin{table}[t]
\caption{Relative cost of reorganization over query ($\alpha$) ranges from 60$\times$ to 100$\times$ in our setup.}
\begin{center}
\footnotesize
\begingroup
\setlength{\tabcolsep}{2pt} % Default value: 6pt
\renewcommand{\arraystretch}{1.2} % Default value: 1
\begin{tabular}{r r r r r r}

    \toprule
     & \multicolumn{5}{c}{\textbf{File Size (MB)}} \\
    \cmidrule{2-6} 
     & 16  & 64  & 256  & 1024  & 4096  \\
     \midrule
        Query (sec)&0.36$\pm$0.05  & 0.89$\pm$0.03  & 2.9$\pm$0.1  & 12.5$\pm$0.2  & 81.0$\pm$2.0 \\
        Reorg (sec)&24.6$\pm$0.7  & 70.0$\pm$4.3  & 276.6$\pm$4.1 & 1231.3$\pm$43.9 & 4854.1$\pm$88.9\\
        $\alpha$&69.0 & 78.7 & 95.4 & 98.4 & 59.9 \\
    \bottomrule

\end{tabular}
\endgroup
\label{tab:alpha}
\vspace{-1.5em}
\end{center}
\end{table}

 \begin{figure}
    \centering
    \includegraphics[width=\columnwidth]{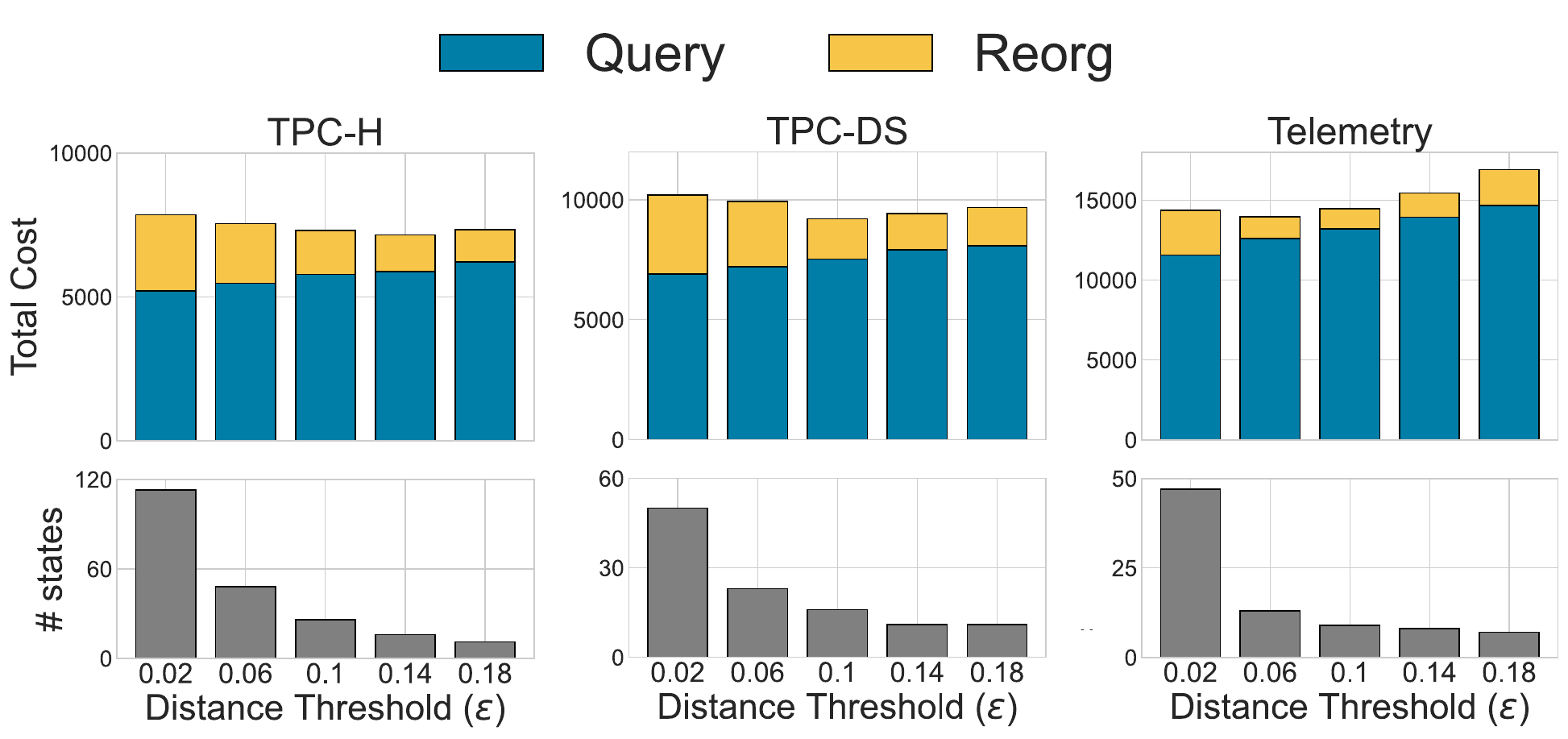}
    \caption{Impact of distance threshold ($\epsilon$) for admitting new layouts. Overall, the framework's performance is not very sensitive to specific choices of the $\epsilon$. }
    \label{fig:eps}
    \vspace{-1.5em}
\end{figure}

\begin{table}
    \caption{Impact of the transition distribution ($\gamma$), use of sliding window (SW) versus reservoir sampling (RS) for candidate data layout generation, and the impact of reorganization delay ($\Delta$) on the MTS algorithm. The rows in bold represent the default experiment parameter configuration. All results are reported in logical costs from simulation in unites of $10^3$. Changes above $5\%$ have been marked in the table.  }
    \footnotesize
      \centering
      \begingroup
        \setlength{\tabcolsep}{3.5pt} % Default value: 6pt
        \renewcommand{\arraystretch}{1.2} % Default value: 1
        \begin{tabular}{r l l l c l l l}
        \toprule
         & \multicolumn{3}{c}{\textbf{Query Cost}} & & \multicolumn{3}{c}{\textbf{Reorg Cost}}\\
      \cmidrule{2-4}  \cmidrule{6-8}
          & TPCH & TPCDS & Telemetry & & TPCH & TPCDS & Telemetry  \\
         \midrule
         \textbf{$\gamma$=1} & 5.56 & 7.39 & 12.60  && 1.68 & 2.24  & 1.52  \\
         $\gamma=0$ & 5.75  & 7.49 & 12.60 && 2.32  & 3.04  &  1.84 \\
          &  & &   && (+38\%) & (+36\%) & (+21\%)  \\
         $\gamma=2$ & 5.56 & 7.39 & 12.60  && 1.68  & 2.24   &  1.60 \\
           &  & &  &&   &   & (+5.3\%)\\
         $\gamma=3$ & 5.56 & 7.39 & 12.56 && 1.68  & 2.16  & 1.52 \\
        \midrule
         \textbf{SW} & 5.56 & 7.39 & 12.60  && 1.68 & 2.24  & 1.52 \\
         RS & 6.51 & 9.03 & 14.66 && 2.00 & 2.16 & 2.24 \\
           & (+17\%) & (+22\%) & (+16\%)  && (+19\%) &  & (+47\%) \\
         SW+RS & 5.59 & 7.19   & 12.55   && 2.40 & 3.04 & 1.44  \\
          &  &   &   && (+43\%) & (+36\%) &  (-5.3\%) \\
        \midrule
        \textbf{$\Delta$=0} & 5.56 & 7.39 & 12.60  && 1.68 & 2.24  & 1.52 \\
        $\Delta$=$40$ & 5.88  & 7.65 & 12.67  && 1.68 & 2.24  & 1.52 \\
          & (+5.7\%)  &  &   &&  &  &  \\
        $\Delta$=$80$ & 6.20 & 7.89 & 12.75  && 1.68 & 2.24  & 1.52\\
         & (+12\%) & (+6.8\%) &   &&   &  & \\
        %$\Delta$=2$\alpha$ &  & 8.41 & 12.88   && 1.68 & 2.24  & 1.52 \\
        %    &  & (+14\%) &    &&   &  &  \\
        \bottomrule
        \end{tabular}
    \endgroup
    \label{tab:gamma}
\end{table}

% \minihead{Effect of distance threshold $\epsilon$}
\subsubsection{Effect of distance threshold $\epsilon$}
Figure~\ref{fig:eps} reports the effect of changing the distance threshold $\epsilon$ on the size of the dynamic state space and the overall performance of the framework.  As $\epsilon$ increases, the size of the state space shrinks and we also observe a slight increase in query cost. The overall performance of the framework is not very sensitive to choices of the distance threshold $\epsilon$, which makes it easy for users to set default parameters for the framework. 

% \minihead{Effect of non-uniform transition distribution $\gamma$} 
\subsubsection{Effect of non-uniform transition distribution $\gamma$}
Table~\ref{tab:gamma} reports the effect of improving the transition distribution with a non-uniform transition distribution enabled by a state performance predictor (\S~\ref{sec:learned}). Specifically, the parameter $\gamma$ controls how much the transition distribution favors states that performed well in the last ``phase" of the algorithm. The original MTS algorithm uses a uniform transition distribution ($\gamma=0$). Overall, using a biased distribution ($\gamma >$  0) improves the reorganization costs of the randomized algorithm by 17.3\% to 27.6\% but does not have a significant impact on the query costs. The performance is not very sensitive to specific choices of $\gamma$. 

% \minihead{Sliding window vs reservoir sampling} 
\subsubsection{Sliding window vs reservoir sampling}
\sysname's layout manager continuously generates data layout candidates using a sliding window (SW) of recent queries. We evaluate the effect of alternative workload sampling strategies such as reservoir sampling (RS) and a combination of candidates from sliding window and reservoir sampling (SW+RS). Our results in Table~\ref{tab:gamma} shows that the use of reservoir sampling led to an increase in query costs by up to 22\% and reorganization costs by up to 47\% when compared to the sliding window strategy. In addition, while the combined strategy (SW+RS) exhibited similar query costs to using a sliding window alone, the reorganization costs can increase by up to 43\%.

% \minihead{Effect of reorganization delay $\Delta$} 
\subsubsection{Effect of reorganization delay $\Delta$}
\sysname performs reorganization in the background, meaning that while a new layout is being created, some queries may still be using the old layout. We study the effect of the delay in the background reorganization on performance by changing the number of queries that are executed using the outdated layout each time the reorganizer decides to switch layouts. The delay does not impact the cost of the reorganization, as the cost is incurred as soon as the decision is made. However, longer delays lead to increased query costs, since the query savings do not take effect until the actual layout switch happens. The query costs increase by around 7 to 12\% compared to the case where there is no delay when the number of queries served on outdated layouts equals $\alpha$. In practice, this latency can be further decreased by dedicating more compute resources for reorganization.

\section{Related Work}
\label{sec:related}

%In this section, we discuss related work in data layout optimization and automatic database tuning.

%\minihead{Data Layout Optimization} 
\subsubsection{Data Layout Optimization}
Data layouts are mapping functions that assign each record in the dataset to different partitions. Partitions are often stored as individual files using compressed, columnar format on local disks or in remote storage. %Data layouts directly impact query performance, especially for selective queries. For example, if each partition only contains a small range of possible values in the query dimensions, many partitions can be skipped when running a selective query over the dataset, thus improving the performance. 

Traditional layout designs such as round-robin, range, and hash partitioning use mapping functions that are independent of both the data distribution and the query workload~\cite{dewitt1992parallel,larson2013enhancements}. More recently, researchers have started to explore specialized layouts that purposefully overfit the layout design to specific datasets and workloads to achieve superior data skipping performance~\cite{kraska2021towards,nathan2020learning,li2020lisa}. Fine-grained workload information such as query predicates are often used in these works to tightly couple the layout designs with target query workloads~\cite{sun2014fine,yang2020qd,ding2021instance,sudhir2023pando}. Our framework leverages these recent developments in workload-aware data layouts as black boxes to generate a list of candidate data layouts that the systems can switch between. In fact, the better these workload-aware layouts perform, the more query cost we can potentially save by dynamically switching between them. \sysname is agnostic to the specific partitioning mechanism used, as long as it can produce different data layouts given different target query workloads. %To show that our framework can work across different partitioning schemes, we generate data layouts using both Qd-trees~\cite{yang2020qd} and Z-ordering in the evaluation.  
%For example, Z-ordering using user-specified attributes improved the fraction of data that can be skipped compared to simply applying a global sort order~\cite{armbrust2020delta}. Using more fine-grained workload information such as query predicates~\cite{sun2014fine,yang2020qd,ding2021instance}, recent works tightly couple the layout design with target query workloads and therefore are able to achieve superior data skipping performance. 

However, workload-aware layout designs can experience significant performance degradation when the target query workload changes. We provide an example in Appendix A of the technical report~\cite{tr}, which shows that a static layout results in almost no savings under changing workloads. Recent work has proposed strategies such as modeling and embedding the variance of workloads into the partitioning design to improve its robustness to workload drifts~\cite{lipaw}. Instead of improving the layout design, our work takes an orthogonal approach and explores how to make better use of existing layouts. Many more opportunities remain in the design and implementation of such systems that can self-optimize and adapt to changes to workload and data~\cite{sdc}.
%, including for non-tabular data such as videos~\cite{daum2021tasm}. 

%\minihead{Automatic Tuning in Databases} 
\subsubsection{Automatic Tuning in Databases}
The online layout optimization problem is an instance of a broader class of problems that aim to automatically tune the physical design and configurations to improve the performance of database systems~\cite{agrawal2005database,chaudhuri2000rethinking,zilio2004db2,chaudhuri2007self,van2017automatic,pavlo2021make}. We focus the discussion below on online settings in which the system adapts the tuning based on changes in the query workload. 

One promising approach is to make tuning decisions according to predictions of future behaviours. For example, researchers 
have used supervised learning to predict future workload patterns and the performance of different system configurations~\cite{ma2018query,akdere2012learning,ding2019ai,mahgoub2020optimuscloud}, reinforcement learning to model system behavior through experiences collected from interacting with the environment~\cite{durand2018gridformation,10.1145/3329859.3329876,learnadvisor}, \revision{or make assumptions on the distribution of query workloads~\cite{ding2021instance}. For example, MTO~\cite{ding2021instance} designs a reward function under the assumption that a certain number of additional queries from the same distribution can be executed before the workload switches, and searches for the reorganization strategy with the maximal reward via dynamic programming. In constrast, \sysname uses online algorithms that neither assume prior knowledge of the query workload nor a specific workload distribution.}
\revision{Another approach is to design rules or heuristics for reorganization conditions, such as current industry practices discussed in \S~\ref{sec:scenario}. For example, SAT~\cite{xie2023sat} monitors the ratio of the actual query selectivity and the data skipping rate, and triggers reorganization process when the ratio is a below certain threshold. While \sysname is inspired by rule-based designs, it introduces a formal framework with worst-case performance guarantees. }
%Either way, there is usually an offline phase where the model is trained based on explicit labels or through trial-and-error, and an online phase where the learned models are fine-tuned and used in the decision making. In contrast, \sysname uses online algorithms, which do not have an offline training phase and are especially useful where historical workloads are infeasible or costly to access.  
%As discussed in \S~\ref{sec:scenario}, \sysname's design is motivated by current rule-based dynamic reorganization practices, which are especially useful where historical workloads are infeasible or costly to access. 

%Another approach is to use online algorithms, which do not rely on predictions of future behaviors and offer the additional benefit of providing worst-case performance guarantees. The closest prior works in this space leverage online algorithms for MTS to adaptively create and delete in-memory indices for database systems as query workload changes~\cite{bruno2007online,schnaitter2010semi,malik2009adaptive}, which we discuss in detail below.
%As discussed in \S~\ref{sec:dss}, the problem setup in adaptive index tuning has a different setup from us since it assumes asymmetric movement cost and a fixed state space. 

%Since the difference here relies on understanding nuances of MTS, we defer the comparison of index tuning and layout optimization to \S~\ref{sec:dss}.  

\subsubsection{Adaptive Index Tuning using Online Algorithms}
%\minihead{Adaptive Index Tuning using Online Algorithms}
Dynamic layout optimization is closely related to a line of work that leverages online algorithms for index tuning and recommendation in database systems~\cite{bruno2007online,schnaitter2010semi,malik2009adaptive}. In indexing tuning, one wants to adaptively create and delete in-memory indices to reduce query execution costs based on workload characteristics. %select and maintain an appropriate set of indexes

%The two problems may seem similar at first glance, but 
Adaptive index tuning has a different problem setup since it assumes asymmetric movement costs and a fixed state space. In index tuning, (nonclustered) indexes store pointers to data records and therefore, maintaining multiple indices requires one copy of the data with some additional storage for the indexing data structures. In contrast, data layouts determine how the dataset is actually sorted and stored and {\em realizing} multiple data layouts require storing multiple copies of the dataset; this should not be confused with {\em evaluating} the query costs on multiple data layouts, which can be estimated using partition-level metadata without accessing the underlying dataset. As a result, in index tuning, there is only the cost of creating additional indexes, and little to no cost with the ``changing" of states, introducing an \emph{asymmetric} movement cost. This difference allows algorithms in our model to achieve a much lower competitive ratio, as uniform metrics are much easier to analyze than asymmetric costs. For example, the {\sc WFIT} algorithm~\cite{schnaitter2010semi} achieves a competitive ratio exponential in the number of possible indexes, whereas our ratio grows logarithmically. In addition, in all prior works on index tuning, the theoretical analysis assumes that the state space of possible indices on the table is fixed.

Metrical task systems can still be applied for asymmetric movement costs, albeit with a worse competitive ratio. The original paper of Borodin \etal~\cite{BLS92} provides an $O(|S|^2)$-competitive algorithm for asymmetric costs following the triangle inequality. In~\cite{bruno2007online}, a 3-competitive algorithm for two-state MTS with asymmetric costs was presented as a special case. We include a proof of an improved competitive ratio for the classic algorithm~\cite{BLS92} in this special case in Appendix C~\cite{tr}.
\section{Discussion and Future Work}
\label{sec:discuss}
%Our work marks an encouraging step towards applying online algorithms in dynamic data layout optimization. 
We outline key directions for future exploration.

First, our study primarily investigates single-table data layouts, but \sysname is also compatible with multi-table configurations. In such setups, each table can maintain its own instance of \sysname and make decisions based on a subset of query predicates relevant to the table. We report preliminary results in Appendix B of the technical report~\cite{tr}, which show that multi-table layouts that utilize predicates induced from joins~\cite{ding2021instance,dipSri2019} show greater benefits from dynamic reorganization compared to layouts that optimize each table separately.

Second, \sysname is based on results from uniform metrical task systems, which assume uniform switching costs between any pair of states. Extending our framework to support non-uniform metrics would increase the possible state space of data layouts. However, algorithms for non-uniform metrics in MTS~\cite{FRT04,FM03,BCLL21,CL19,EL22} are considerably more complex. Adapting them to the dynamic variant we introduced in this paper presents an interesting venue for future work.

\revision{Third, \sysname does not keep additional copies of the data with different layouts, except temporarily during reorganization. However, if we have the storage budget to maintain multiple layouts of the same dataset simultaneously, we could extend \Cref{alg:update-states} to accommodate this scenario. A variant of our algorithm tailored for this purpose is introduced in Appendix D of the technical report~\cite{tr}. Further analysis of the optimal tradeoffs involved in maintaining multiple layouts, as well as exploration of alternative variants of our algorithm, could be pursued in an extended work.}

%Second, we focus on scenarios where the time taken for reorganization is relatively short compared to the rate at which query patterns change, such as in the SuperCollider query workload, whose query patterns remain stable over short periods (like days) but can shift over longer spans (like months). In scenarios where query patterns change rapidly, new data layouts might become outdated by the time the reorganization is finished. Further analysis of the performance impact of query arrival rate and workload stability would be a valuable extension to our framework.

% \section{Discussion and Future Work}
% First, non-uniform metric for reorganization cost. 

% Second, modeling query arrival rate .

% Third, extending to multiple tables. 

\section{Conclusion}
\label{sec:conclude}
We introduce \sysname, an algorithmic framework that decides when and how to change the data layout to minimize overall data access and movement over the entire query stream, without prior knowledge of the query workload. \sysname leverages recent developments in workload-aware data layout designs and offers a systematic way to reason about the trade off between query and reorganization costs when utilizing such layout designs. \sysname makes reorganization decisions in an online fashion by extending classic results for metrical task systems to support dynamic state space and achieves a tight competitive ratio that is logarithmic in the maximum size of the dynamic state space. Our empirical findings show that dynamic reorganization using \sysname achieves sizable speed ups in end-to-end query and reorganization time compared to using a single, optimized data layout throughout, and compares favorably with oracles that observe the entire query stream in advance and switch states dynamically.  

\section{Acknowledgment}
Research reported in this work was supported by an Amazon Research Award Fall 2023. Any opinions, findings, and conclusions or recommendations expressed in this material are those of the author(s) and do not reflect the views of Amazon. Moses Charikar was supported by a Simons investigator award.

\bibliographystyle{IEEEtran}
\bibliography{IEEEabrv,dumts}

\begin{thebibliography}{10}
\providecommand{\url}[1]{#1}
\csname url@samestyle\endcsname
\providecommand{\newblock}{\relax}
\providecommand{\bibinfo}[2]{#2}
\providecommand{\BIBentrySTDinterwordspacing}{\spaceskip=0pt\relax}
\providecommand{\BIBentryALTinterwordstretchfactor}{4}
\providecommand{\BIBentryALTinterwordspacing}{\spaceskip=\fontdimen2\font plus
\BIBentryALTinterwordstretchfactor\fontdimen3\font minus \fontdimen4\font\relax}
\providecommand{\BIBforeignlanguage}[2]{{%
\expandafter\ifx\csname l@#1\endcsname\relax
\typeout{** WARNING: IEEEtran.bst: No hyphenation pattern has been}%
\typeout{** loaded for the language `#1'. Using the pattern for}%
\typeout{** the default language instead.}%
\else
\language=\csname l@#1\endcsname
\fi
#2}}
\providecommand{\BIBdecl}{\relax}
\BIBdecl

\bibitem{graefe2009fast}
G.~Graefe, ``Fast loads and fast queries,'' in \emph{International Conference on Data Warehousing and Knowledge Discovery}.\hskip 1em plus 0.5em minus 0.4em\relax Springer, 2009, pp. 111--124.

\bibitem{oraclezm}
``{Oracle Database Data Warehousing Guide: Using Zone Maps},'' \url{https://docs.oracle.com/en/database/oracle/oracle-database/21/dwhsg/using-zone-maps.html}, 2020, accessed October 2023.

\bibitem{sparkskipping}
``{Data skipping with Z-order indexes for Delta Lake},'' \url{https://docs.databricks.com/en/delta/data-skipping.html}, 2023, accessed October 2023.

\bibitem{sparksql}
\BIBentryALTinterwordspacing
M.~Armbrust, R.~S. Xin, C.~Lian, Y.~Huai, D.~Liu, J.~K. Bradley, X.~Meng, T.~Kaftan, M.~J. Franklin, A.~Ghodsi, and M.~Zaharia, ``Spark {SQL:} relational data processing in spark,'' in \emph{Proceedings of the 2015 {ACM} {SIGMOD} International Conference on Management of Data, Melbourne, Victoria, Australia, May 31 - June 4, 2015}, T.~K. Sellis, S.~B. Davidson, and Z.~G. Ives, Eds.\hskip 1em plus 0.5em minus 0.4em\relax {ACM}, 2015, pp. 1383--1394. [Online]. Available: \url{https://doi.org/10.1145/2723372.2742797}
\BIBentrySTDinterwordspacing

\bibitem{borthakur2011apache}
D.~Borthakur, J.~Gray, J.~S. Sarma, K.~Muthukkaruppan, N.~Spiegelberg, H.~Kuang, K.~Ranganathan, D.~Molkov, A.~Menon, S.~Rash \emph{et~al.}, ``Apache hadoop goes realtime at facebook,'' in \emph{Proceedings of the 2011 ACM SIGMOD International Conference on Management of data}, 2011, pp. 1071--1080.

\bibitem{sun2014fine}
L.~Sun, M.~J. Franklin, S.~Krishnan, and R.~S. Xin, ``Fine-grained partitioning for aggressive data skipping,'' in \emph{Proceedings of the 2014 ACM SIGMOD international conference on Management of data}, 2014, pp. 1115--1126.

\bibitem{yang2020qd}
Z.~Yang, B.~Chandramouli, C.~Wang, J.~Gehrke, Y.~Li, U.~F. Minhas, P.-{\AA}. Larson, D.~Kossmann, and R.~Acharya, ``Qd-tree: Learning data layouts for big data analytics,'' in \emph{Proceedings of the 2020 ACM SIGMOD International Conference on Management of Data}, 2020, pp. 193--208.

\bibitem{ding2021instance}
J.~Ding, U.~F. Minhas, B.~Chandramouli, C.~Wang, Y.~Li, Y.~Li, D.~Kossmann, J.~Gehrke, and T.~Kraska, ``Instance-optimized data layouts for cloud analytics workloads,'' in \emph{Proceedings of the 2021 International Conference on Management of Data}, 2021, pp. 418--431.

\bibitem{snowauto}
``{Automatic Clustering},'' \url{https://docs.snowflake.com/en/user-guide/tables-auto-reclustering}, 2024, accessed February 2024.

\bibitem{deltaauto}
``{Configure Delta Lake to control data file size},'' \url{https://docs.databricks.com/en/delta/tune-file-size.html}, December 2023, accessed Feburary 2024.

\bibitem{BLS92}
\BIBentryALTinterwordspacing
A.~Borodin, N.~Linial, and M.~E. Saks, ``An optimal on-line algorithm for metrical task system,'' \emph{J. {ACM}}, vol.~39, no.~4, pp. 745--763, 1992. [Online]. Available: \url{https://doi.org/10.1145/146585.146588}
\BIBentrySTDinterwordspacing

\bibitem{snow}
R.~Shelly, ``{Automatic Clustering at Snowflake},'' \url{https://medium.com/snowflake/automatic-clustering-at-snowflake-317e0bb45541}, January 2022, accessed February 2024.

\bibitem{borodin2005online}
A.~Borodin and R.~El-Yaniv, \emph{Online computation and competitive analysis}.\hskip 1em plus 0.5em minus 0.4em\relax cambridge university press, 2005.

\bibitem{IS98}
\BIBentryALTinterwordspacing
S.~Irani and S.~S. Seiden, ``Randomized algorithms for metrical task systems,'' \emph{Theor. Comput. Sci.}, vol. 194, no. 1-2, pp. 163--182, 1998. [Online]. Available: \url{https://doi.org/10.1016/S0304-3975(97)00006-6}
\BIBentrySTDinterwordspacing

\bibitem{rong2020approximate}
\BIBentryALTinterwordspacing
K.~Rong, Y.~Lu, P.~Bailis, S.~Kandula, and P.~A. Levis, ``Approximate partition selection for big-data workloads using summary statistics,'' \emph{Proc. {VLDB} Endow.}, vol.~13, no.~11, pp. 2606--2619, 2020. [Online]. Available: \url{http://www.vldb.org/pvldb/vol13/p2606-rong.pdf}
\BIBentrySTDinterwordspacing

\bibitem{morton1966computer}
G.~M. Morton, ``A computer oriented geodetic data base and a new technique in file sequencing,'' 1966.

\bibitem{armbrust2020delta}
M.~Armbrust, T.~Das, L.~Sun, B.~Yavuz, S.~Zhu, M.~Murthy, J.~Torres, H.~van Hovell, A.~Ionescu, A.~{\L}uszczak \emph{et~al.}, ``Delta lake: high-performance acid table storage over cloud object stores,'' \emph{Proceedings of the VLDB Endowment}, vol.~13, no.~12, pp. 3411--3424, 2020.

\bibitem{liquid}
``{Use liquid clustering for Delta tables},'' \url{https://docs.databricks.com/en/delta/clustering.html}, February 2024.

\bibitem{bestpractice}
``{Data partitioning guidance},'' \url{https://learn.microsoft.com/en-us/azure/architecture/best-practices/data-partitioning}, accessed February, 2024.

\bibitem{cohen2017online}
\BIBentryALTinterwordspacing
A.~Cohen and S.~Mannor, ``Online learning with many experts,'' \emph{CoRR}, vol. abs/1702.07870, 2017. [Online]. Available: \url{http://arxiv.org/abs/1702.07870}
\BIBentrySTDinterwordspacing

\bibitem{hentschel2019general}
B.~Hentschel, P.~J. Haas, and Y.~Tian, ``General temporally biased sampling schemes for online model management,'' \emph{ACM Transactions on Database Systems (TODS)}, vol.~44, no.~4, pp. 1--45, 2019.

\bibitem{sudhir2023pando}
S.~Sudhir, W.~Tao, N.~Laptev, C.~Habis, M.~Cafarella, and S.~Madden, ``Pando: Enhanced data skipping with logical data partitioning,'' \emph{Proceedings of the VLDB Endowment}, vol.~16, no.~9, pp. 2316--2329, 2023.

\bibitem{daum2021tasm}
M.~Daum, B.~Haynes, D.~He, A.~Mazumdar, and M.~Balazinska, ``Tasm: A tile-based storage manager for video analytics,'' in \emph{2021 IEEE 37th International Conference on Data Engineering (ICDE)}.\hskip 1em plus 0.5em minus 0.4em\relax IEEE, 2021, pp. 1775--1786.

\bibitem{dewitt1992parallel}
D.~DeWitt and J.~Gray, ``Parallel database systems: The future of high performance database systems,'' \emph{Communications of the ACM}, vol.~35, no.~6, pp. 85--98, 1992.

\bibitem{larson2013enhancements}
P.-A. Larson, C.~Clinciu, C.~Fraser, E.~N. Hanson, M.~Mokhtar, M.~Nowakiewicz, V.~Papadimos, S.~L. Price, S.~Rangarajan, R.~Rusanu \emph{et~al.}, ``Enhancements to sql server column stores,'' in \emph{Proceedings of the 2013 ACM SIGMOD International Conference on Management of Data}, 2013, pp. 1159--1168.

\bibitem{kraska2021towards}
T.~Kraska, ``Towards instance-optimized data systems,'' \emph{Proceedings of the VLDB Endowment}, vol.~14, no.~12, 2021.

\bibitem{nathan2020learning}
V.~Nathan, J.~Ding, M.~Alizadeh, and T.~Kraska, ``Learning multi-dimensional indexes,'' in \emph{Proceedings of the 2020 ACM SIGMOD international conference on management of data}, 2020, pp. 985--1000.

\bibitem{li2020lisa}
P.~Li, H.~Lu, Q.~Zheng, L.~Yang, and G.~Pan, ``Lisa: A learned index structure for spatial data,'' in \emph{Proceedings of the 2020 ACM SIGMOD international conference on management of data}, 2020, pp. 2119--2133.

\bibitem{tr}
``{Dynamic Data Layout Optimization with Worst-case Guarantees},'' \url{https://github.com/d2i-lab/oreo/blob/master/docs/tr.pdf}, accessed November, 2023.

\bibitem{lipaw}
Z.~Li, M.~L. Yiu, and T.~N. Chan, ``Paw: Data partitioning meets workload variance,'' \emph{ICDE}, 2022.

\bibitem{sdc}
S.~Madden, J.~Ding, T.~Kraska, S.~Sudhir, D.~Cohen, T.~Mattson, and N.~Tatbul, ``Self-organizing data containers,'' in \emph{CIDR}, 2022.

\bibitem{agrawal2005database}
S.~Agrawal, S.~Chaudhuri, L.~Kollar, A.~Marathe, V.~Narasayya, and M.~Syamala, ``Database tuning advisor for microsoft sql server 2005,'' in \emph{Proceedings of the 2005 ACM SIGMOD international conference on Management of data}, 2005, pp. 930--932.

\bibitem{chaudhuri2000rethinking}
S.~Chaudhuri and G.~Weikum, ``Rethinking database system architecture: Towards a self-tuning risc-style database system.'' in \emph{VLDB}, 2000, pp. 1--10.

\bibitem{zilio2004db2}
D.~C. Zilio, J.~Rao, S.~Lightstone, G.~Lohman, A.~Storm, C.~Garcia-Arellano, and S.~Fadden, ``Db2 design advisor: Integrated automatic physical database design,'' in \emph{Proceedings of the Thirtieth international conference on Very large data bases-Volume 30}, 2004, pp. 1087--1097.

\bibitem{chaudhuri2007self}
S.~Chaudhuri and V.~Narasayya, ``Self-tuning database systems: a decade of progress,'' in \emph{Proceedings of the 33rd international conference on Very large data bases}, 2007, pp. 3--14.

\bibitem{van2017automatic}
D.~Van~Aken, A.~Pavlo, G.~J. Gordon, and B.~Zhang, ``Automatic database management system tuning through large-scale machine learning,'' in \emph{Proceedings of the 2017 ACM international conference on management of data}, 2017, pp. 1009--1024.

\bibitem{pavlo2021make}
A.~Pavlo, M.~Butrovich, L.~Ma, P.~Menon, W.~S. Lim, D.~Van~Aken, and W.~Zhang, ``Make your database system dream of electric sheep: towards self-driving operation,'' \emph{Proceedings of the VLDB Endowment}, vol.~14, no.~12, pp. 3211--3221, 2021.

\bibitem{ma2018query}
L.~Ma, D.~Van~Aken, A.~Hefny, G.~Mezerhane, A.~Pavlo, and G.~J. Gordon, ``Query-based workload forecasting for self-driving database management systems,'' in \emph{Proceedings of the 2018 International Conference on Management of Data}, 2018, pp. 631--645.

\bibitem{akdere2012learning}
\BIBentryALTinterwordspacing
M.~Akdere, U.~{\c{C}}etintemel, M.~Riondato, E.~Upfal, and S.~B. Zdonik, ``Learning-based query performance modeling and prediction,'' in \emph{{IEEE} 28th International Conference on Data Engineering {(ICDE} 2012), Washington, DC, {USA} (Arlington, Virginia), 1-5 April, 2012}, A.~Kementsietsidis and M.~A.~V. Salles, Eds.\hskip 1em plus 0.5em minus 0.4em\relax {IEEE} Computer Society, 2012, pp. 390--401. [Online]. Available: \url{https://doi.org/10.1109/ICDE.2012.64}
\BIBentrySTDinterwordspacing

\bibitem{ding2019ai}
B.~Ding, S.~Das, R.~Marcus, W.~Wu, S.~Chaudhuri, and V.~R. Narasayya, ``Ai meets ai: Leveraging query executions to improve index recommendations,'' in \emph{Proceedings of the 2019 International Conference on Management of Data}, 2019, pp. 1241--1258.

\bibitem{mahgoub2020optimuscloud}
A.~Mahgoub, A.~M. Medoff, R.~Kumar, S.~Mitra, A.~Klimovic, S.~Chaterji, and S.~Bagchi, ``$\{$OPTIMUSCLOUD$\}$: Heterogeneous configuration optimization for distributed databases in the cloud,'' in \emph{2020 USENIX Annual Technical Conference (USENIX ATC 20)}, 2020, pp. 189--203.

\bibitem{durand2018gridformation}
G.~C. Durand, M.~Pinnecke, R.~Piriyev, M.~Mohsen, D.~Broneske, G.~Saake, M.~S. Sekeran, F.~Rodriguez, and L.~Balami, ``Gridformation: towards self-driven online data partitioning using reinforcement learning,'' in \emph{Proceedings of the First International Workshop on Exploiting Artificial Intelligence Techniques for Data Management}, 2018, pp. 1--7.

\bibitem{10.1145/3329859.3329876}
\BIBentryALTinterwordspacing
B.~Hilprecht, C.~Binnig, and U.~R\"{o}hm, ``Towards learning a partitioning advisor with deep reinforcement learning,'' ser. aiDM '19.\hskip 1em plus 0.5em minus 0.4em\relax New York, NY, USA: Association for Computing Machinery, 2019. [Online]. Available: \url{https://doi.org/10.1145/3329859.3329876}
\BIBentrySTDinterwordspacing

\bibitem{learnadvisor}
\BIBentryALTinterwordspacing
B.~Hilprecht, C.~Binnig, and U.~Rohm, ``Learning a partitioning advisor for cloud databases,'' in \emph{Proceedings of the 2020 ACM SIGMOD International Conference on Management of Data}, ser. SIGMOD '20.\hskip 1em plus 0.5em minus 0.4em\relax New York, NY, USA: Association for Computing Machinery, 2020, p. 143–157. [Online]. Available: \url{https://doi.org/10.1145/3318464.3389704}
\BIBentrySTDinterwordspacing

\bibitem{xie2023sat}
X.~Xie, S.~Shi, H.~Wang, and M.~Li, ``Sat: sampling acceleration tree for adaptive database repartition,'' \emph{World Wide Web}, vol.~26, no.~5, pp. 3503--3533, 2023.

\bibitem{bruno2007online}
N.~Bruno and S.~Chaudhuri, ``An online approach to physical design tuning,'' in \emph{2007 IEEE 23rd International Conference on Data Engineering}.\hskip 1em plus 0.5em minus 0.4em\relax IEEE, 2007, pp. 826--835.

\bibitem{schnaitter2010semi}
\BIBentryALTinterwordspacing
K.~Schnaitter and N.~Polyzotis, ``Semi-automatic index tuning: Keeping dbas in the loop,'' \emph{Proc. VLDB Endow.}, vol.~5, no.~5, p. 478–489, jan 2012. [Online]. Available: \url{https://doi.org/10.14778/2140436.2140444}
\BIBentrySTDinterwordspacing

\bibitem{malik2009adaptive}
T.~Malik, X.~Wang, D.~Dash, A.~Chaudhary, A.~Ailamaki, and R.~Burns, ``Adaptive physical design for curated archives,'' in \emph{International Conference on Scientific and Statistical Database Management}.\hskip 1em plus 0.5em minus 0.4em\relax Springer, 2009, pp. 148--166.

\bibitem{dipSri2019}
\BIBentryALTinterwordspacing
S.~Kandula, L.~Orr, and S.~Chaudhuri, ``Pushing data-induced predicates through joins in big-data clusters,'' \emph{Proc. VLDB Endow.}, vol.~13, no.~3, p. 252–265, nov 2019. [Online]. Available: \url{https://doi.org/10.14778/3368289.3368292}
\BIBentrySTDinterwordspacing

\bibitem{FRT04}
J.~Fakcharoenphol, S.~Rao, and K.~Talwar, ``A tight bound on approximating arbitrary metrics by tree metrics,'' \emph{Journal of Computer and System Sciences}, vol.~69, no.~3, pp. 485--497, 2004.

\bibitem{FM03}
A.~Fiat and M.~Mendel, ``Better algorithms for unfair metrical task systems and applications,'' \emph{SIAM Journal on Computing}, vol.~32, no.~6, pp. 1403--1422, 2003.

\bibitem{BCLL21}
S.~Bubeck, M.~B. Cohen, J.~R. Lee, and Y.~T. Lee, ``Metrical task systems on trees via mirror descent and unfair gluing,'' \emph{SIAM Journal on Computing}, vol.~50, no.~3, pp. 909--923, 2021.

\bibitem{CL19}
C.~Coester and J.~R. Lee, ``Pure entropic regularization for metrical task systems,'' in \emph{Conference on Learning Theory}.\hskip 1em plus 0.5em minus 0.4em\relax PMLR, 2019, pp. 835--848.

\bibitem{EL22}
F.~Ebrahimnejad and J.~R. Lee, ``Multiscale entropic regularization for mts on general metric spaces,'' in \emph{13th Innovations in Theoretical Computer Science Conference (ITCS 2022)}.\hskip 1em plus 0.5em minus 0.4em\relax Schloss Dagstuhl-Leibniz-Zentrum f{\"u}r Informatik, 2022.

\end{thebibliography}

%\appendix
%\input{sections/appendix}

\end{document}